\documentclass[a4paper,12pt]{article}         
\usepackage{amsmath,amsfonts,amsthm,amssymb}  
\usepackage[utf8]{inputenc} 
\usepackage[T1]{fontenc}    
									   
\usepackage{graphicx,caption}            
\usepackage{color}                       
\usepackage[hidelinks]{hyperref}                    
\hypersetup{breaklinks=true}

\usepackage{breakurl}
\usepackage{grffile}
\usepackage{makeidx}                  
\usepackage{acronym}
\usepackage{mathrsfs}
\usepackage{float}
\usepackage{vruler}                    
\usepackage{palatino, url, multicol}
\usepackage{setspace}                  
\usepackage{txfonts}
\usepackage{booktabs}
\usepackage{bm}

\usepackage{tikz}
\usepackage{tikz-3dplot}
\usetikzlibrary{shapes.geometric, arrows, calc, 3d}

\usepackage[numbers,sort&compress]{natbib}   
\usepackage{import}

\usepackage{caption}
\usepackage{subcaption}
\usepackage{listings}
\usepackage{comment}
\usepackage[title]{appendix}

\usepackage{array}
\usepackage{ragged2e}
\newcolumntype{P}[1]{>{\RaggedRight\arraybackslash}p{#1}}

\newtheorem{theorem}{Theorem}[section]
\newtheorem{proposition}[theorem]{Proposition}
\newtheorem{corollary}[theorem]{Corollary}
\newtheorem{lemma}[theorem]{Lemma}

\newtheorem{definition}[theorem]{Definition}

\restylefloat{figure}

\title{Surface-Encoded Partial Coherence Transformation: \\ Modeling Source Coherence Effects in Wave Optics}

\author{Netzer Moriya} 

\date{}

\begin{document}

\maketitle

\begin{abstract}
We present a new mathematical framework for incorporating partial coherence effects into wave optics simulations 
through a comprehensive surface-to-detector approach. Unlike traditional ensemble averaging methods, our dual-component 
framework models partial coherence through: (1) a surface-encoded transformation implemented via a linear integral 
operator with a spatially-dependent kernel that modifies coherence properties at the reflection interface, followed by (2) 
a propagation component that evolves these coherence properties to the detection plane. This approach differs fundamentally 
from conventional models by explicitly separating surface interactions from propagation effects, while maintaining a 
unified mathematical structure. 

We derive the mathematical foundation based on the coherence function formalism, establish the connection to the 
Van Cittert-Zernike theorem, and prove the equivalence of our framework to conventional partial coherence theory. 
The method reduces the dimensional complexity of coherence calculations and offers potential computational advantages, 
particularly for systems involving multiple surfaces and propagation steps. Applications include optical testing 
and astronomical instrumentation. We provide rigorous mathematical proofs, demonstrate the convergence properties, 
and analyze the relative importance of surface and propagation effects across different optical scenarios.
\end{abstract}

\section{Introduction}

The accurate modeling of partial coherence in optical systems has remained an enduring challenge in computational optics. 
While the behavior of fully coherent wavefronts is well understood within the framework of scalar diffraction theory and 
electromagnetic wave propagation, the incorporation of statistical properties inherent in real-world sources introduces 
significant complexity. This challenge is most acute in imaging, interferometry, and optical testing systems where 
source coherence directly influences resolution, contrast, and fringe visibility.

The classical theory of partial coherence is rooted in the mutual coherence function \( \Gamma(\bm{r}_1, \bm{r}_2, \tau) \), 
which encodes the second-order correlations between field amplitudes at two space-time points. In the case of stationary, 
quasi-monochromatic sources, the mutual coherence reduces to a spatially dependent function \( \gamma(\bm{r}_1, \bm{r}_2) \), 
which is the central object of interest in most coherence modeling scenarios.

One of the most foundational results in this context is the Van Cittert–Zernike theorem (VCZ), originally derived for radio 
sources by Van Cittert and later extended by Zernike to optical systems~\cite{vanCittert1934, zernike1938, mandel_wolf}. 
According to the VCZ theorem, 
the spatial coherence of an optical field in the far field (or under paraxial approximation) is the 
Fourier transform of the source's intensity distribution. Specifically, for an incoherent planar source with 
intensity \( I_s(\bm{\rho}) \), the mutual coherence function in the observation plane at distance \( z \) is given by

\begin{equation}
\Gamma(\bm{r}_1, \bm{r}_2) \propto \mathcal{F}\left[ I_s(\bm{\rho}) \right] \left( \frac{\bm{r}_2 - \bm{r}_1}{\lambda z} \right),
\end{equation}

where \( \mathcal{F} \) denotes the Fourier transform and \( \lambda \) the wavelength. This establishes a deep link between 
source structure and coherence in the propagated field.

Numerous extensions and applications of the VCZ theorem have been developed. It underpins stellar interferometry, as in 
the classic Michelson and Hanbury Brown–Twiss experiments \cite{michelson1921, hanburybrown1956}, where coherence measurements are used to infer angular diameters 
of stars. It also plays a foundational role in optical metrology and partially coherent imaging, such as in the modeling of 
Gaussian-Schell model beams~\cite{wolf_josa, goodman}. However, all these methods treat coherence as an emergent property 
due to propagation from an incoherent or partially coherent source.

Despite their power, conventional coherence modeling techniques face significant limitations in computational efficiency 
and conceptual clarity. Traditional approaches to modeling partial coherence generally fall into three categories, each with distinct limitations:

\begin{enumerate}
    \item \textbf{Ensemble averaging methods} require generating multiple realizations of coherent fields with appropriate statistics, then averaging the results. While physically intuitive, this approach is computationally expensive, requiring many realizations to achieve convergence, and scales poorly for complex systems.
    
    \item \textbf{Direct propagation of the mutual coherence function} involves propagating the four-dimensional coherence function $\Gamma(\bm{r}_1, \bm{r}_2)$ through the optical system. While mathematically rigorous, this approach is computationally intensive due to the high-dimensional nature of the problem and often lacks closed-form expressions for complex systems.
    
    \item \textbf{Phase-space methods}, such as the Wigner distribution approach, represent partially coherent light in a joint position-momentum space. While elegant mathematically, these methods often introduce interpretability challenges and computational complexity.
\end{enumerate}

These limitations become particularly pronounced in complex optical systems where coherence effects must be tracked through multiple surfaces and propagation steps.

Several advanced frameworks have been developed to address these limitations. The Wigner function approach, pioneered by 
Alonso and Wolf~\cite{Alonso2011}, offers a complete phase-space description but often introduces computational 
complexity and interpretability challenges. The Mutual Intensity Transport Equation (MITE)~\cite{FischerVisser2004} 
provides an elegant differential formulation but can be difficult to solve numerically for complex systems. Coherent mode 
decomposition methods~\cite{Gori2003} enable the use of coherent propagation techniques for each mode but 
may require numerous modes to accurately represent certain coherence states. In contrast to these approaches, which 
primarily focus on propagation through volumes, our proposed framework offers a distinctive perspective by explicitly separating 
surface interactions from propagation effects, while providing a unified mathematical framework that encompasses both.

In this paper, we introduce a fundamentally different modeling paradigm: the Surface-Encoded Coherence Transformation (SECT) 
framework consisting of two complementary components:

\begin{enumerate}
    \item \textbf{Surface Component (SECT$_\mathbf{S}$)}: A surface transformation component that models how optical interfaces modify coherence properties through a surface-encoded coherence operator $\mathcal{C}_S$.
    
    \item \textbf{Propagation Component (SECT$_\mathbf{P}$)}: A propagation component that accounts for the evolution of these coherence properties during subsequent propagation to the detection plane through a propagation operator $\mathcal{P}_z$.
\end{enumerate}

This dual-component approach allows us to explicitly separate the physical mechanisms that influence coherence while maintaining a unified mathematical framework. It provides both a new conceptual understanding of coherence evolution and potential computational advantages for complex optical systems.

The first component, SECT$_S$, is inspired by the mathematical structure of the VCZ 
theorem. It interprets coherence modification as a boundary interaction effect governed by an operator \( \mathcal{C}_S \), 
acting at the surface of reflection or transmission \cite{friberg1982gsm}. This surface-localized transformation allows a fully coherent field to be transformed into a partially coherent field at the reflecting surface.

The second component, SECT$_P$, accounts for how the coherence properties established at the surface 
evolve during propagation to the detection plane. This component captures the well-established physics of coherence propagation 
while integrating seamlessly with the surface transformation.

This approach offers several advantages. Conceptually, it provides a clear separation of surface and propagation effects, offering insight into their relative contributions in different optical scenarios. Mathematically, it provides a unified framework that reduces the dimensional complexity of coherence calculations. Computationally, it enables more efficient simulation of partially coherent systems by combining the advantages of surface-localized transformations with established propagation methods.

\vspace{0.3cm}

Our objectives in this paper are as follows:
\begin{enumerate}
    \item We define and construct a mathematically rigorous operator formalism for coherence transformation at surfaces and during subsequent propagation.
    \item We demonstrate that, under specific conditions, our framework reproduces known coherence behavior consistent with VCZ predictions and establish its formal equivalence to traditional approaches.
    \item We analyze the theoretical properties of both components (SECT$_S$ and SECT$_P$), including their connections to spatial filtering theory and natural extensions to spectral and temporal coherence.
    \item We establish the theoretical limits and validity domains of our approach through mathematical analysis of convergence and boundary conditions.
    \item We provide a comparative analysis of the relative importance of surface and propagation effects across different optical scenarios, offering practical guidelines for applications.
\end{enumerate}

To clarify the conceptual and computational distinctions of our framework relative to established methods, 
Table~\ref{table:Comparison_of_SECT_with_existing_approaches} 
summarizes the key differences between SECT and three representative approaches: the Mutual Intensity Transport 
Equation (MITE), Wigner function formalism, and coherent mode decomposition. This comparison highlights the specific 
structural advantages of the surface-encoded formulation and its implications for modeling partially coherent fields in 
complex systems.

\begin{table}[ht]
\centering
\caption{Comparison of SECT with existing partial coherence modeling approaches}
\label{table:Comparison_of_SECT_with_existing_approaches}
\resizebox{\textwidth}{!}{%
\begin{tabular}{|l|P{3.2cm}|P{3cm}|P{3cm}|P{3cm}|}
\hline
\textbf{Method} & \textbf{Core Representation} & \textbf{Assumptions} & \textbf{Computational Complexity} & \textbf{Notes} \\
\hline
SECT (this work) & Operator-based, 2-step (surface + propagation) & Scalar, paraxial, quasi-monochromatic & \( \mathcal{O}(N^2 \log N) \) per surface & Explicit separation of surface and propagation effects \\
\hline
MITE & Differential transport of mutual coherence function & High-dimensional coherence function \( \Gamma(r_1, r_2) \) & \( \mathcal{O}(N^4) \) & Rigorous, but difficult to solve numerically \\
\hline
Wigner function formalism & Phase-space distribution & Stationary, paraxial, Gaussian preferred & \( \mathcal{O}(N^3) \) typical & Elegant, but reduced physical interpretability \\
\hline
Coherent mode decomposition & Sum of orthogonal coherent modes & Positive-semidefinite mutual coherence & \( \mathcal{O}(M N^2) \), where \( M \) modes & Efficient when few modes suffice; basis-dependent \\
\hline
\end{tabular}
}
\end{table}

Beyond general-purpose wave optics simulations, the Partial Coherence Transformation (PCT) framework finds natural 
relevance in astronomical imaging \cite{buscher2002}, where extended or structured celestial sources often exhibit partial spatial coherence. 
Applications such as long-baseline optical interferometry, adaptive optics correction, and high-resolution aperture synthesis 
routinely encounter coherence effects that influence both signal interpretation and instrument design. 
The ability of the proposed method to encode spatial coherence directly on a surface—rather than relying solely on 
statistical ensemble 
models—offers a deterministic approach to simulate and propagate coherence-aware fields across complex optical systems. 
This capability could be particularly valuable for next-generation astronomical observatories aiming to model partially 
coherent fields from stars, disks, and other astrophysical phenomena with greater fidelity.

By shifting the paradigm from purely source-based or propagation-based coherence modeling to our dual-component framework, we offer a new lens through which partial coherence can be understood theoretically and implemented computationally. This approach offers a complementary perspective to existing methods, particularly valuable for analyzing coherence effects in complex optical systems with multiple surfaces and propagation paths.

\section{Theoretical Framework}

\subsection{Coherence Function Formalism}

The state of partial coherence of a quasi-monochromatic scalar optical field is fully described by the mutual coherence 
function, defined as~\cite{mandel_wolf, goodman, born1999principles}:
\begin{equation}
\Gamma(\bm{r}_1, \bm{r}_2, \tau) = \left\langle U^*(\bm{r}_1, t) U(\bm{r}_2, t + \tau) \right\rangle,
\end{equation}
where \( U(\bm{r}, t) \) is the complex scalar optical field at position \( \bm{r} \) and time \( t \), and the angle 
brackets \( \langle \cdot \rangle \) denote an ensemble average over the statistical realizations of the field. 
Throughout this work, we assume that \( U \) is a scalar field; that is, polarization effects are neglected unless 
stated otherwise.

For statistically stationary, quasi-monochromatic sources, and under the assumption that spatial and temporal 
correlations are separable, the mutual coherence function can be factorized as~\cite{mandel_wolf}:
\begin{equation}
\Gamma(\bm{r}_1, \bm{r}_2, \tau) = \gamma(\bm{r}_1, \bm{r}_2)\, \phi(\tau),
\label{eq:cross_spectral_density}
\end{equation}
where \( \gamma(\bm{r}_1, \bm{r}_2) \) represents the spatial coherence function and \( \phi(\tau) \) is the 
temporal coherence function. This separation implies that the spatial and temporal coherence properties of the 
field evolve independently—a valid approximation for many practical sources but not universally applicable\footnote{In broadband or strongly chirped fields, the spatial degree of coherence may itself vary with frequency,
so \(\Gamma\) must be treated in its general two-frequency form \(W(\bm r_1,\bm r_2;\omega,\omega')\) rather than 
by Eq.~\ref{eq:cross_spectral_density}}.
The assumption of statistical stationarity ensures that the mutual coherence function depends only on the time
difference \( \tau \), not on the absolute time \( t \), and separability requires that the joint space-time
coherence structure is the product of two independent functions.

Throughout this work, we employ the paraxial approximation in our propagation formalism. This approximation is valid when the following condition is satisfied:
\begin{equation}
\alpha \ll 1
\end{equation}
where $\alpha$ represents the maximum angle between any ray and the optical axis. Quantitatively, this requires that 
the maximum lateral extent of both the reflecting surface and the detection plane be much smaller than the propagation 
distance:
\begin{equation} \label{eq:maximum_propagation_distance_condition}
\max(L_{\text{surface}}, L_{\text{detector}}) \ll z
\end{equation}
where $L_{\text{surface}}$ and $L_{\text{detector}}$ represent the characteristic lateral dimensions of the surface and 
detector, respectively, and $z$ is the propagation distance. 

When this condition is not satisfied, the paraxial approximation breaks down, and the assumptions underlying the 
Fresnel kernel used in our propagation operator \( P_z \) no longer hold. In such cases, more general scalar diffraction 
theories must be used to account for the exact geometry and propagation angles. 
These include the Rayleigh-Sommerfeld \cite{goodman2005introduction} or 
angular spectrum methods, as well as vectorial formulations for high-numerical-aperture systems such as those developed by 
Richards and Wolf~\cite{RichardsWolf1959}. Although Section~\ref{sec:Propagation_Operator} introduces the exact propagation kernel for completeness, the present work is 
limited to the paraxial regime, and the analysis of non-paraxial propagation remains outside the scope of this framework.

\subsection{Van Cittert-Zernike Theorem Extension}

\subsubsection{Equivalence to Van Cittert-Zernike Theorem} \label{sec:brief_outline_of_the_equivalence_proof}

We now establish the mathematical equivalence between our dual-component framework and the 
traditional Van Cittert-Zernike (VCZ) approach under appropriate conditions. This equivalence is central to 
validating our approach.

\begin{theorem}[SECT-VCZ Equivalence] \label{thrm:SECT-VCZ_Equivalence}
Given a fully coherent incident field $U_i(\bm{r}) = A_0$ (uniform plane wave illumination) impinging on a 
surface with coherence operator $\mathcal{C}_S$ characterized by kernel $K_S(\bm{r}, \bm{r}', \lambda)$, followed by 
propagation operator $\mathcal{P}_z$, the mutual coherence function at the detection plane is mathematically equivalent 
to that obtained from the VCZ theorem for an incoherent source with intensity distribution $I_s(\bm{\rho})$, 
provided that\footnote{This proportional form reflects the spatial coherence structure at the surface as derived from the Van Cittert–Zernike 
theorem under paraxial conditions. The exact normalization and dimensional scaling of the Fourier transform are 
carried out in Appendix~\ref{sup:SECT-VCZ}.}:
\begin{equation}
K_S(\bm{r}, \bm{r}', \lambda) \propto \mathcal{F}\left\{I_s\left(\frac{\bm{\rho}}{\lambda z}\right)\right\}(\bm{r} - \bm{r}')
\label{eq:SECT-VCZ_Equivalence}
\end{equation}
where $\mathcal{F}$ denotes the Fourier transform operator.
\end{theorem}

\begin{proof}
A brief outline of the proof is as follows: We start with the mutual coherence function at the detection plane using our 
SECT framework and express it in terms of the field after application of the surface coherence operator and propagation. 
By representing the surface field as having stochastic fluctuations around a mean value, we derive its second-order 
statistics. After propagation using the Fresnel kernel and application of the convolution theorem, we obtain a form 
mathematically equivalent to the Van Cittert-Zernike theorem. A detailed derivation with all mathematical steps is 
provided in Appendix~\ref{sup:SECT-VCZ}.
\end{proof}

The kernel condition given in Equation~\ref{eq:SECT-VCZ_Equivalence}, which appears in the statement of 
Theorem~\ref{thrm:SECT-VCZ_Equivalence}, holds under the same physical assumptions as the Van Cittert–Zernike theorem. 
Specifically, the source must be spatially incoherent, and the propagation from the source to the surface must occur in 
the paraxial regime. The equivalence assumes that the angular distribution of the incoherent source maps to the 
spatial coherence function at the surface via a Fourier transform. This requires the surface to be located in the 
far field of the source, consistent with the condition \( \max(L_\text{surface}, L_\text{detector}) \ll z \) given 
in Equation~\ref{eq:maximum_propagation_distance_condition}. 
Under these assumptions, the surface coherence kernel \( K_S(r, r', \lambda) \) defined in 
Theorem~\ref{thrm:SECT-VCZ_Equivalence} reproduces the mutual coherence function predicted by the VCZ theorem.

For practical implementation, we now derive explicit forms for two common source distributions that are widely used 
in optical systems:

\begin{corollary}[Circular Source Equivalence]
For a uniform circular incoherent source of angular radius $\theta$, the equivalent surface coherence kernel is:
\begin{equation}
K_S(\bm{r} - \bm{r}', \lambda) = \frac{2J_1(\pi\theta|\bm{r} - \bm{r}'|/\lambda)}{\pi\theta|\bm{r} - \bm{r}'|/\lambda}
\end{equation}
where $J_1$ is the first-order Bessel function of the first kind.
\end{corollary}

\begin{proof}
For a uniform circular source with intensity distribution:
\begin{align}
I_s(\bm{\rho}) = 
\begin{cases}
I_0 & \text{if } |\bm{\rho}| \leq \theta \\
0 & \text{otherwise}
\end{cases}
\end{align}

The Fourier transform required for our kernel is:
\begin{align}
\mathcal{F}\left\{I_s\left(\frac{\bm{\rho}}{\lambda z}\right)\right\}(\bm{r}_1 - \bm{r}_2) &= \int_{|\bm{\rho}| \leq \theta} I_0 e^{-2\pi i \bm{\rho} \cdot (\bm{r}_1 - \bm{r}_2)/\lambda z} \frac{d^2\rho}{(\lambda z)^2}
\end{align}

Converting to polar coordinates with $\bm{\rho} = \rho (\cos\phi, \sin\phi)$ and setting $\bm{\Delta r} = \bm{r}_1 - \bm{r}_2 = |\bm{\Delta r}|(\cos\alpha, \sin\alpha)$, we have:
\begin{align}
\mathcal{F}\left\{I_s\left(\frac{\bm{\rho}}{\lambda z}\right)\right\}(\bm{\Delta r}) &= \frac{I_0}{(\lambda z)^2} \int_0^{\theta} \rho d\rho \int_0^{2\pi} e^{-2\pi i \rho |\bm{\Delta r}| \cos(\phi-\alpha)/\lambda z} d\phi \\
&= \frac{2\pi I_0}{(\lambda z)^2} \int_0^{\theta} \rho J_0\left(\frac{2\pi \rho |\bm{\Delta r}|}{\lambda z}\right) d\rho \notag
\end{align}

where we've used the integral representation of the zeroth-order Bessel function: $J_0(x) = \frac{1}{2\pi}\int_0^{2\pi} e^{ix\cos\phi} d\phi$.

Evaluating the remaining integral using the property $\int_0^a x J_0(bx) dx = \frac{a}{b}J_1(ab)$, we obtain:
\begin{align}
\mathcal{F}\left\{I_s\left(\frac{\bm{\rho}}{\lambda z}\right)\right\}(\bm{\Delta r}) &= \frac{2\pi I_0}{(\lambda z)^2} \frac{\lambda z}{2\pi|\bm{\Delta r}|} \theta J_1\left(\frac{2\pi \theta |\bm{\Delta r}|}{\lambda z}\right) \\
&= \frac{I_0 \theta^2}{\lambda z} \frac{2J_1\left(\frac{\pi \theta |\bm{\Delta r}|}{\lambda}\right)}{\frac{\pi \theta |\bm{\Delta r}|}{\lambda}} \notag
\end{align}

After normalization, this gives the surface coherence kernel for a circular source:
\begin{align}
K_S(\bm{r} - \bm{r}', \lambda) = \frac{2J_1(\pi\theta|\bm{r} - \bm{r}'|/\lambda)}{\pi\theta|\bm{r} - \bm{r}'|/\lambda}
\end{align}

This is the well-known Airy pattern form of the mutual coherence function, which characterizes the spatial coherence from a circular incoherent source of angular diameter $2\theta$.
\end{proof}

\begin{corollary}[Gaussian Source Equivalence]
For a Gaussian incoherent source with RMS angular width $\sigma_\theta$, the equivalent surface coherence kernel is:
\begin{equation}
K_S(\bm{r} - \bm{r}', \lambda) = \exp\left(-\frac{\pi^2\sigma_\theta^2|\bm{r} - \bm{r}'|^2}{2\lambda^2}\right)
\end{equation}
\end{corollary}

\begin{proof}
For a Gaussian source with intensity distribution:
\begin{align}
I_s(\bm{\rho}) = I_0 \exp\left(-\frac{|\bm{\rho}|^2}{2\sigma_\theta^2}\right)
\end{align}

The Fourier transform required for our kernel is:
\begin{align}
\mathcal{F}\left\{I_s\left(\frac{\bm{\rho}}{\lambda z}\right)\right\}(\bm{\Delta r}) &= \int I_0 \exp\left(-\frac{|\bm{\rho}|^2}{2\sigma_\theta^2}\right) e^{-2\pi i \bm{\rho} \cdot \bm{\Delta r}/\lambda z} \frac{d^2\rho}{(\lambda z)^2}
\end{align}

Using the Fourier transform property that the transform of a Gaussian is another Gaussian, specifically $\mathcal{F}\{e^{-\pi a x^2}\} = \frac{1}{\sqrt{a}}e^{-\pi \xi^2/a}$, and applying appropriate scaling, we obtain:
\begin{align}
\mathcal{F}\left\{I_s\left(\frac{\bm{\rho}}{\lambda z}\right)\right\}(\bm{\Delta r}) &= \frac{I_0 \cdot 2\pi\sigma_\theta^2}{(\lambda z)^2} \exp\left(-\frac{\pi^2\sigma_\theta^2|\bm{\Delta r}|^2}{2\lambda^2}\right)
\end{align}

After normalization, this gives the surface coherence kernel for a Gaussian source:
\begin{align}
K_S(\bm{r} - \bm{r}', \lambda) = \exp\left(-\frac{\pi^2\sigma_\theta^2|\bm{r} - \bm{r}'|^2}{2\lambda^2}\right)
\end{align}

This is the Gaussian form of the mutual coherence function, which leads to the widely used Gaussian-Schell model in coherence theory. The coherence length is inversely proportional to the angular size of the source, in accordance with the van Cittert-Zernike theorem.
\end{proof}

These explicit forms for common source distributions provide practical implementation guidance and demonstrate that our 
framework reproduces well-established coherence patterns. The circular source case gives rise to the familiar Airy-disk 
coherence pattern common in astronomical observations through circular apertures (see in Section~\ref{appendix:Stellar_Coherence_Transformation}), while the Gaussian source produces the 
mathematically convenient Gaussian coherence function used extensively in coherence theory.

\subsection{Two-Component Formulation}

One of the central features of this framework is the explicit separation of partial coherence modeling into two distinct physical components, each represented by its own mathematical operator. This separation provides both physical insight and computational advantages over traditional approaches that treat coherence primarily as a source-based or propagation-based phenomenon.

Figure~\ref{fig:System_Geometry_with_Curved_Surfaces} illustrates the geometric configuration of our framework, showing a 
curved reflecting surface and detection plane with incident and reflected rays. The diagram visualizes how the 
SECT$_S$ component encodes coherence at the reflecting surface, while the SECT$_P$ component governs the propagation of 
these coherence properties along the path to the detection plane.

\begin{figure}[H]
    \centering
    \includegraphics[width=0.9\textwidth]{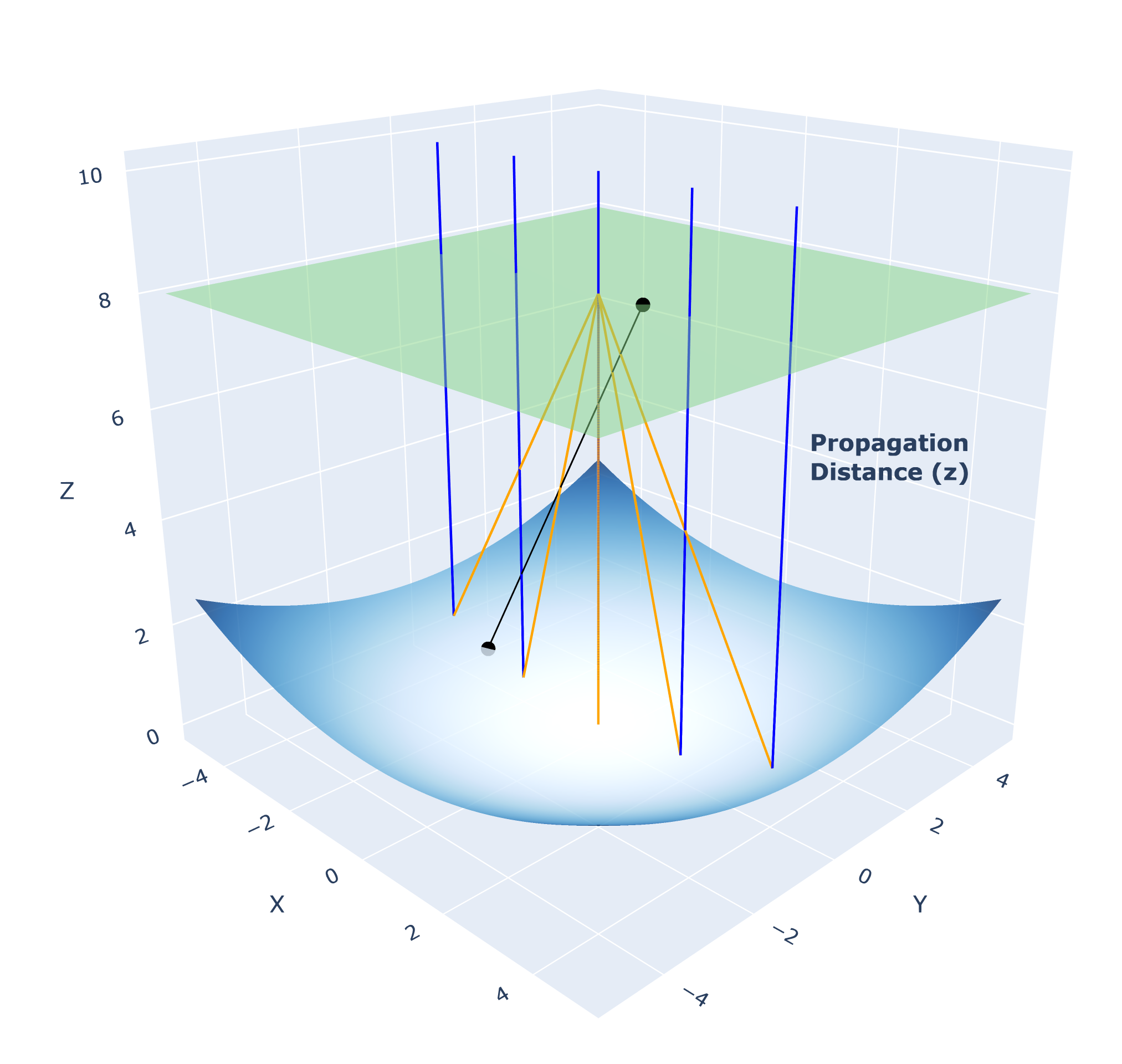}
    \caption{Geometric configuration of the SECT framework showing a curved reflecting surface and detection plane. 
The framework accounts for the actual three-dimensional coordinates of each point on both surfaces connected by black line
to illustrate the point-to-point propagation distance.}
    \label{fig:System_Geometry_with_Curved_Surfaces}
\end{figure}

\subsubsection{Component 1: Surface-Encoded Coherence Transformation (SECT$_S$)}

The first component, denoted as SECT$_S$, represents the transformation of coherence properties that occurs at the reflecting or transmitting surface. We define this component through a surface coherence operator $\mathcal{C}_S$ that acts on the incident field $U_i(\bm{r})$ immediately after the standard reflection or transmission operation:

\begin{equation}
U_s(\bm{r}) = \mathcal{C}_S(\mathcal{R}(U_i))(\bm{r}) = \int K_S(\bm{r}, \bm{r}', \lambda) \mathcal{R}(U_i)(\bm{r}') d^2r'
\end{equation}

where $\mathcal{R}$ is the reflection operator, $K_S(\bm{r}, \bm{r}', \lambda)$ is the surface coherence kernel that 
encodes the desired coherence properties, and $U_s(\bm{r})$ is the field at the surface 
after both reflection and coherence transformation.

In the case of spatially stationary kernels, the integral operator \( C_S \) corresponds to a spatial convolution, 
allowing an interpretation of the SECT$_S$ component as a spatial low-pass filter applied to the reflected field. 
This filtering perspective will be further developed in Section~\ref{sec:Relation_to_Spatial_Filtering_Theory}.

This surface coherence operator effectively transforms a fully coherent incident field into a partially coherent 
field directly at the surface, mimicking the effect of an extended source. The key physical insight is that 
we can model the partial coherence as being "imprinted" on the wavefront at the surface, rather than requiring an 
ensemble of source realizations.

\subsubsection{Component 2: Propagation Component (SECT$_P$)}

The second component, denoted as SECT$_P$, accounts for how the coherence properties established at the surface evolve during propagation to the detection plane. We define this component through a propagation operator $\mathcal{P}_z$ that acts on the field after it has been transformed by the surface coherence operator:

\begin{equation}
U_d(\bm{r}) = \mathcal{P}_z(U_s)(\bm{r}) = \int h(\bm{r}, \bm{r}', z, \lambda) U_s(\bm{r}') d^2r'
\end{equation}

where $h(\bm{r}, \bm{r}', z, \lambda)$ is the propagation kernel (typically the Fresnel or Rayleigh-Sommerfeld kernel), $z$ is the propagation distance, and $U_d(\bm{r})$ is the field at the detection plane.

The propagation component captures the evolution of coherence properties during free-space propagation, including effects such as diffraction and the natural increase in coherence length with distance (as described by the Van Cittert-Zernike theorem).

It is important to note that the Fresnel propagation kernel is derived under the paraxial approximation, which assumes 
that the dominant propagation occurs along the optical axis with small angular deviations. This approximation limits 
the validity of our framework to systems where the propagation distance $z$ is significantly larger than the characteristic 
lateral dimensions of both the surface and detection plane. For high-numerical-aperture systems or extreme off-axis 
configurations, a more general non-paraxial propagation kernel would be required.

\subsection{Generalization to Curved Surfaces and Arbitrary Detection Geometries}

While the formulation presented thus far assumes planar surfaces and detection planes oriented perpendicular to the optical axis, the framework naturally extends to more general geometries. For a curved reflecting surface described by a height function $h(\bm{r})$, the reflection operator $\mathcal{R}$ must account for the local surface normal:
\begin{equation}
\mathcal{R}(U)(\bm{r}) = R(\bm{r})e^{i\phi_s(\bm{r})}U(\bm{r})
\end{equation}
where $\phi_s(\bm{r}) = 2k \cdot h(\bm{r})$ for a reflective surface under normal incidence. More generally, for arbitrary incidence angles, the phase term must account for the optical path difference based on the local surface normal.

For propagation to an arbitrary detection surface described by coordinate $\bm{r}_d$, the propagation kernel must account for the actual point-to-point distance:
\begin{equation}
h(\bm{r}_d, \bm{r}, \lambda) = \frac{1}{i\lambda |\bm{r}_d - \bm{r}|} e^{ik|\bm{r}_d - \bm{r}|}
\end{equation}
where $|\bm{r}_d - \bm{r}|$ is the Euclidean distance between a point $\bm{r}$ on the reflecting surface and a point $\bm{r}_d$ on the detection surface. This generalized propagation kernel reduces to the Fresnel approximation when both the surface and detector are planar and the paraxial approximation holds.

The mutual coherence function at the detection surface then becomes:
\begin{equation}
\Gamma_d(\bm{r}_{d1}, \bm{r}_{d2}) = \int\int h^*(\bm{r}_{d1}, \bm{r}'_1, \lambda) h(\bm{r}_{d2}, \bm{r}'_2, \lambda) \Gamma_S(\bm{r}'_1, \bm{r}'_2) d^2r'_1 d^2r'_2
\end{equation}

While this formulation is mathematically complete, practical computation often requires additional approximations, such 
as the stationary phase approximation or numerical quadrature methods, particularly when the paraxial approximation is 
not valid.

The SECT framework relies on the sequential application of two distinct mathematical components: the surface-encoded 
coherence transformation (SECT$_S$), which modifies coherence at the interface via a kernel operator \( C_S \), and the 
propagation component (SECT$_P$), which evolves the field to the detection plane through \( P_z \). 
Each of these has been defined individually in the preceding sections. We now describe how these components 
integrate to form the complete two-step transformation from an incident field \( U_i \) to the detected field \( U_d \). 
This composition forms the basis for a unified operator framework in which partial coherence is introduced 
deterministically at the surface and propagated through standard diffraction theory. The resulting formulation 
preserves the interpretability and modularity of the two-step approach, while enabling efficient and physically accurate 
simulations for a wide range of optical configurations.

\subsubsection{Complete Two-Component Framework}

The complete framework combines these two components sequentially, providing a comprehensive description of how coherence evolves from the surface to the detection plane:

\begin{equation}
U_d(\bm{r}) = \mathcal{P}_z(\mathcal{C}_S(\mathcal{R}(U_i)))(\bm{r})
\end{equation}

The mutual coherence function at the detection plane is then given by:

\begin{equation}
\Gamma_d(\bm{r}_1, \bm{r}_2) = \langle U_d^*(\bm{r}_1) U_d(\bm{r}_2) \rangle
\end{equation}

This dual-component approach differs fundamentally from traditional methods in two important ways:

1. It explicitly separates the physical mechanisms that influence coherence (surface interactions versus propagation effects), providing clearer physical insight into the relative contributions of each.

2. It allows coherence to be modeled through deterministic operators rather than requiring ensemble averaging over multiple realizations, potentially offering significant computational advantages.

In the following subsections, we develop the mathematical properties of each operator in detail and establish their relationship to existing coherence theory.

\subsection{Surface-Encoded Coherence Operator}

The first component of our framework, SECT$_S$, is implemented through the surface-encoded coherence 
operator $\mathcal{C}_S$, defined as a linear integral operator that acts on a field $U(\bm{r})$ according to:
\begin{equation}
\mathcal{C}_S(U)(\bm{r}) = \int K_S(\bm{r}, \bm{r}', \lambda) U(\bm{r}') d^2r'
\end{equation}

where $K_S(\bm{r}, \bm{r}', \lambda)$ is the surface coherence kernel that encapsulates the statistical coherence properties to be encoded at the surface. This operator transforms a fully coherent field into a partially coherent one by introducing appropriate spatial correlations directly at the reflection interface.

\begin{theorem}[Surface-Encoded Coherence Transformation]
The SECT$_S$ component transforms a fully coherent incident field $U_i(\bm{r})$ into a field with the statistical properties of a partially coherent field at the reflecting surface according to:
\begin{equation}
U_s(\bm{r}) = \mathcal{C}_S(\mathcal{R}(U_i))(\bm{r})
\end{equation}
where $\mathcal{R}$ is the standard surface reflection operator defined as:

\begin{equation}
\mathcal{R}(U)(\bm{r}) = R(\bm{r})e^{i\phi_s(\bm{r})}U(\bm{r}) 
\end{equation}

with $R(\bm{r})$ being the reflection coefficient and $\phi_s(\bm{r})$ the phase change due to surface geometry.
\end{theorem}

\begin{proof}
Let $U_i(\bm{r})$ be a fully coherent incident field. The reflected field after surface interaction is:
\begin{equation}
\mathcal{R}(U_i)(\bm{r}) = R(\bm{r}) e^{i\phi_s(\bm{r})} U_i(\bm{r})
\end{equation}

Applying the surface coherence operator:
\begin{align}
U_s(\bm{r}) &= \mathcal{C}_S(\mathcal{R}(U_i))(\bm{r}) \\
&= \int K_S(\bm{r}, \bm{r}', \lambda) R(\bm{r}') e^{i\phi_s(\bm{r}')} U_i(\bm{r}') d^2r' \notag
\end{align}

To verify that this produces the expected statistical properties of a partially coherent field, we examine the mutual coherence function at the surface:
\begin{align}
\Gamma_S(\bm{r}_1, \bm{r}_2) &= \langle U_s^*(\bm{r}_1) U_s(\bm{r}_2) \rangle \\
&= \langle \left(\int K_S^*(\bm{r}_1, \bm{r}'_1, \lambda) R^*(\bm{r}'_1) e^{-i\phi_s(\bm{r}'_1)} U_i^*(\bm{r}'_1) d^2r'_1\right) \notag \\ 
& \quad \times \left(\int K_S(\bm{r}_2, \bm{r}'_2, \lambda) R(\bm{r}'_2) e^{i\phi_s(\bm{r}'_2)} U_i(\bm{r}'_2) d^2r'_2\right) \rangle \notag
\end{align}

For a fully coherent incident field, $\langle U_i^*(\bm{r}'_1) U_i(\bm{r}'_2) \rangle = U_i^*(\bm{r}'_1) U_i(\bm{r}'_2)$. 
Assuming a uniform plane wave illumination $U_i(\bm{r}) = A_0$ and a uniform reflection 
coefficient $R(\bm{r}) = R_0$ for simplicity:
\begin{align}
\Gamma_S(\bm{r}_1, \bm{r}_2) &= |R_0|^2 |A_0|^2 \int\int K_S^*(\bm{r}_1, \bm{r}'_1, \lambda) K_S(\bm{r}_2, \bm{r}'_2, \lambda) \\
&e^{i[\phi_s(\bm{r}'_2) - \phi_s(\bm{r}'_1)]} d^2r'_1 d^2r'_2 \notag
\end{align}

For a flat surface\footnote{The large-scale \emph{figure} of a curved mirror contributes a deterministic
phase that can be subtracted analytically; after this removal the residual
roughness phase $\phi_{\text{rough}}$ is statistically homogeneous, so the
kernel remains spatially stationary.  See \cite{moriya2025rdf} for an explicit
separation of figure and roughness.}, where $\phi_s(\bm r)=\phi_0$ is constant,
and assuming a spatially stationary kernel
$K_S(\bm r,\bm r';\lambda)=K_S(\bm r-\bm r';\lambda)$ with appropriate
normalization, this reduces to:
\begin{align}
\Gamma_S(\bm r_1,\bm r_2)
  &= |R_0|^2 |A_0|^2 K_S(\bm r_1-\bm r_2;\lambda)
  \label{eq:GammaSFlat}
\end{align}

This is exactly the form expected for a partially coherent field at the surface with coherence 
function $\gamma_S(\bm{r}_1, \bm{r}_2) = K_S(\bm{r}_1 - \bm{r}_2, \lambda)$, thus demonstrating that the SECT$_S$ 
component successfully encodes the desired coherence properties at the surface.
\end{proof}

The power of the SECT$_S$ component lies in its ability to transform a deterministic coherent field into a field with the statistical properties of a partially coherent field through a deterministic operator. This avoids the need for ensemble averaging over multiple field realizations, offering potential computational advantages. Furthermore, by encoding coherence properties at the surface, we establish a clear physical interpretation of how surface properties influence the coherence structure of the reflected field.

\subsection{Propagation Operator} \label{sec:Propagation_Operator}

The second component of our framework, SECT$_P$, accounts for the propagation of the partially coherent field from 
the reflecting surface to the detection plane. We define this component through a propagation operator $\mathcal{P}_z$ 
acting as:
\begin{equation}
\mathcal{P}_z(U)(\bm{r}) = \int h(\bm{r}, \bm{r}', z, \lambda) U(\bm{r}') d^2r'
\end{equation}
where $h(\bm{r}, \bm{r}', z, \lambda)$ is the appropriate propagation kernel that governs how the field evolves during 
propagation over distance $z$. 

For paraxial propagation, the Fresnel kernel is used:
\begin{equation}
h_{\text{paraxial}}(\bm{r}, \bm{r}', z, \lambda) = \frac{e^{ikz}}{i\lambda z} e^{i\frac{k}{2z}|\bm{r} - \bm{r}'|^2}
\end{equation}
where $k = 2\pi/\lambda$ is the wavenumber. It is important to note that the Fresnel propagation kernel is derived under the paraxial approximation, which assumes that the dominant propagation occurs along the optical axis with small angular deviations. This approximation limits the validity of our framework to systems where the propagation distance $z$ is significantly larger than the characteristic lateral dimensions of both the surface and detection plane:
\begin{equation}
\max(L_{\text{surface}}, L_{\text{detector}}) \ll z
\end{equation}
where $L_{\text{surface}}$ and $L_{\text{detector}}$ represent the characteristic lateral dimensions of the surface and 
detector, respectively.

Figure~\ref{fig:Exact_vs_Paraxial_Path_Calculation} illustrates the difference between exact path calculation and the paraxial approximation, showing that for systems where the lateral dimensions are much smaller than the propagation distance, the path difference becomes negligible. However, for large angles or curved surfaces, the exact calculation becomes necessary to accurately model coherence evolution.

\begin{figure}[H]
    \centering
    \includegraphics[width=0.9\textwidth]{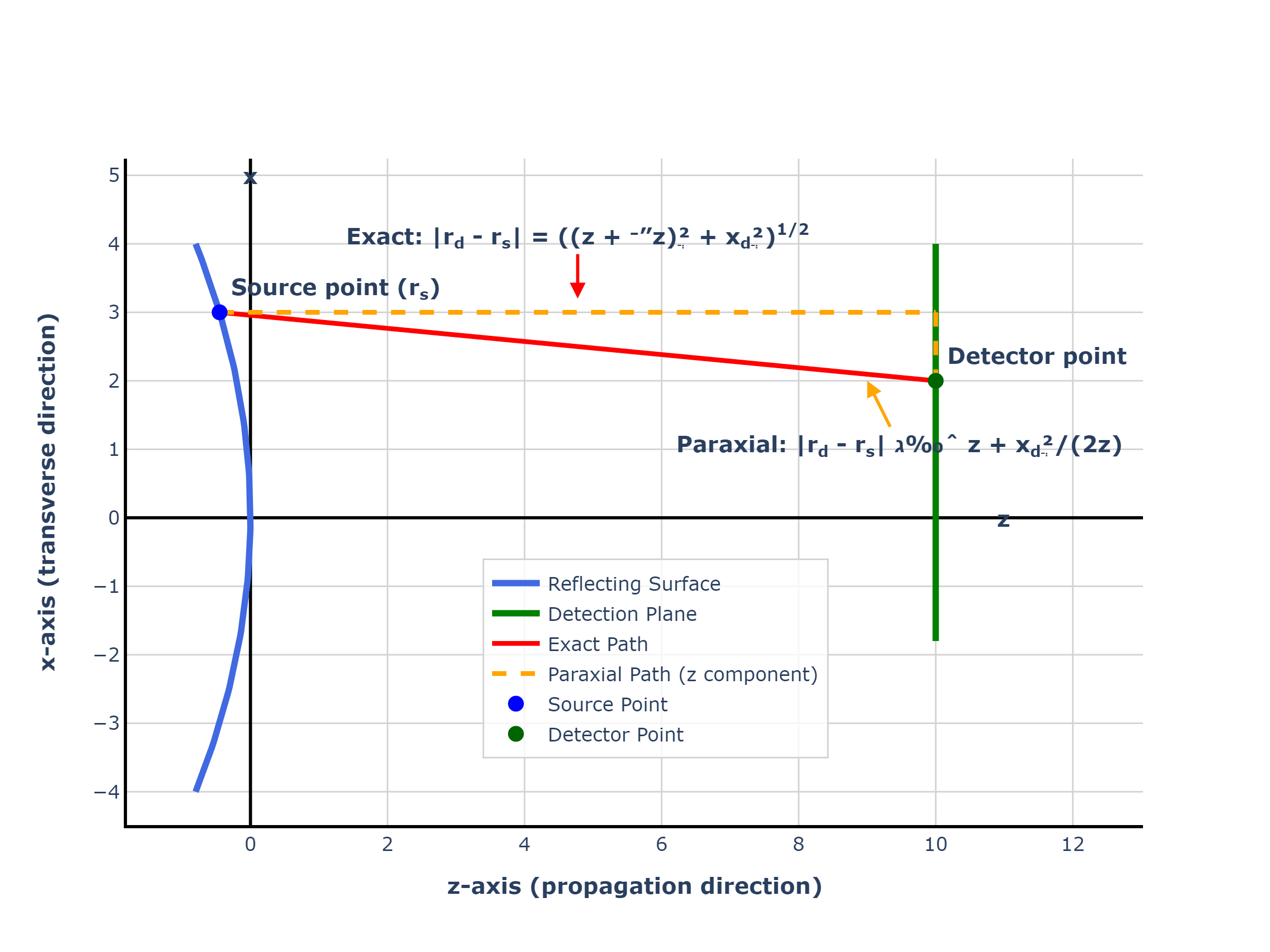}
    \caption{Comparison between exact path calculation (required for curved surfaces and large angles) and the paraxial 
approximation (valid when lateral dimensions are much smaller than propagation distance).}
    \label{fig:Exact_vs_Paraxial_Path_Calculation}
\end{figure}

More generally, for arbitrary geometries including curved surfaces and non-planar detectors, the exact propagation kernel is:
\begin{equation}
h_{\text{exact}}(\bm{r}, \bm{r}', \lambda) = \frac{1}{i\lambda |\bm{r} - \bm{r}'|} e^{ik|\bm{r} - \bm{r}'|}
\end{equation}
where $|\bm{r} - \bm{r}'|$ is the exact Euclidean distance between points. The paraxial approximation is obtained when $|\bm{r} - \bm{r}'| \approx z + \frac{|\bm{r}_{\perp} - \bm{r}'_{\perp}|^2}{2z}$, where $\bm{r}_{\perp}$ and $\bm{r}'_{\perp}$ represent the transverse coordinates.

For a curved reflecting surface described by a height function $h_s(\bm{r})$ and a detection surface that may 
not be perpendicular to the optical axis, the propagation must account for the actual three-dimensional coordinates 
of each point. In this case, the propagation integral becomes:
\begin{equation}
\mathcal{P}(U)(\bm{r}_d) = \int h_{\text{exact}}(\bm{r}_d, \bm{r}_s, \lambda) U(\bm{r}_s) dA_s
\end{equation}
where $\bm{r}_d$ represents coordinates on the detection surface, $\bm{r}_s$ represents coordinates on the reflecting 
surface, and $dA_s$ is the differential area element which may vary across a curved surface.

The field at the detection plane is then given by:
\begin{equation}
U_d(\bm{r}) = \mathcal{P}_z(U_s)(\bm{r}) = \mathcal{P}_z(\mathcal{C}_S(\mathcal{R}(U_i)))(\bm{r})
\end{equation}

This sequential application of the reflection operator $\mathcal{R}$, surface coherence operator $\mathcal{C}_S$ (SECT$_S$), 
and propagation operator $\mathcal{P}_z$ (SECT$_P$) provides a complete description of how the field evolves from incidence 
to detection.

The mutual coherence function at the detection plane is given by:
\begin{equation}
\Gamma_d(\bm{r}_1, \bm{r}_2) = \langle U_d^*(\bm{r}_1) U_d(\bm{r}_2) \rangle
\end{equation}

Unlike the SECT$_S$ component, which introduces coherence modifications, the SECT$_P$ component primarily evolves the 
existing coherence structure according to well-established diffraction principles. For a partially coherent 
field with mutual coherence function $\Gamma_S(\bm{r}_1, \bm{r}_2)$ at the surface, the mutual coherence function 
at the detection plane is:
\begin{equation}
\Gamma_d(\bm{r}_1, \bm{r}_2) = \int\int h^*(\bm{r}_1, \bm{r}'_1, z, \lambda) h(\bm{r}_2, \bm{r}'_2, z, \lambda) \Gamma_S(\bm{r}'_1, \bm{r}'_2) d^2r'_1 d^2r'_2
\end{equation}

For curved surfaces and arbitrary detection geometries, this expression generalizes to:
\begin{equation}
\Gamma_d(\bm{r}_{d1}, \bm{r}_{d2}) = \int\int h_{\text{exact}}^*(\bm{r}_{d1}, \bm{r}'_{s1}, \lambda) h_{\text{exact}}(\bm{r}_{d2}, \bm{r}'_{s2}, \lambda) \Gamma_S(\bm{r}'_{s1}, \bm{r}'_{s2}) dA_{s1} dA_{s2}
\end{equation}
where the integration is performed over the actual surface areas.

This propagation of the mutual coherence function is consistent with the standard treatment in partial coherence theory, 
ensuring that our SECT$_P$ component correctly captures the evolution of coherence during propagation.

A key insight of our framework is that while both SECT$_S$ and SECT$_P$ influence the final coherence properties, 
they do so through distinctly different physical mechanisms: SECT$_S$ establishes the initial coherence structure 
through surface interactions, while SECT$_P$ evolves this structure through diffraction and free-space propagation. 
This separation allows for more targeted modeling and optimization in different optical scenarios, as we will explore 
in Section~\ref{sec:relative_importance}. The flexibility of our framework accommodates both paraxial systems (where 
computational efficiency can 
be maximized) and non-paraxial systems with curved surfaces (where accuracy is paramount), though different mathematical 
and computational techniques may be required for each regime.

\subsection{Combined Formulation}

Having established the individual components of our framework—the surface-encoded coherence operator (SECT$_S$) and the propagation operator (SECT$_P$)—we now develop a unified mathematical formulation that combines these components while preserving their distinct physical interpretations. This combined approach provides computational advantages while maintaining the conceptual clarity of the dual-component structure.

We can combine the surface coherence operator $\mathcal{C}_S$ (SECT$_S$) and propagation operator $\mathcal{P}_z$ (SECT$_P$) into a total coherence operator $\mathcal{C}_{\text{tot}}$ that directly maps from the incident field to the field at the detection plane:

\begin{equation}
U_d(\bm{r}) = \mathcal{C}_{\text{tot}}(\mathcal{R}(U_i))(\bm{r})
\end{equation}

where the total coherence operator is defined as:

\begin{equation}
\mathcal{C}_{\text{tot}}(U)(\bm{r}) = \int K_{\text{tot}}(\bm{r}, \bm{r}', z, \lambda) U(\bm{r}') d^2r'
\end{equation}

The total coherence kernel $K_{\text{tot}}$ represents the combined effect of the SECT$_S$ and SECT$_P$ components and is 
related to 
the surface coherence kernel $K_S$ and the propagation kernel $h$ through:

\begin{equation}
K_{\text{tot}}(\bm{r}, \bm{r}', z, \lambda) = \int h(\bm{r}, \bm{r}'', z, \lambda) K_S(\bm{r}'', \bm{r}', \lambda) d^2r''
\end{equation}

This integration effectively combines the two sequential operations (surface coherence transformation followed by propagation) into a single operation. For computational purposes, this combined formulation can offer efficiency advantages, particularly when multiple field points need to be evaluated at the detection plane.

The mutual coherence function at the detection plane can then be expressed as:

\begin{equation}
\Gamma_d(\bm{r}_1, \bm{r}_2) = \langle U_d^*(\bm{r}_1) U_d(\bm{r}_2) \rangle = \langle [\mathcal{C}_{\text{tot}}(\mathcal{R}(U_i))(\bm{r}_1)]^* \mathcal{C}_{\text{tot}}(\mathcal{R}(U_i))(\bm{r}_2) \rangle
\end{equation}

For a fully coherent incident field $U_i(\bm{r})$ and a deterministic reflection operator $\mathcal{R}$, this simplifies to:

\begin{align}
\Gamma_d(\bm{r}_1, \bm{r}_2) = \int\int K_{\text{tot}}^*(\bm{r}_1, \bm{r}'_1, z, \lambda) K_{\text{tot}}(\bm{r}_2, \bm{r}'_2, z, \lambda) \\
&U_i^*(\bm{r}'_1) U_i(\bm{r}'_2) d^2r'_1 d^2r'_2 \notag
\end{align}

For the common case of uniform plane wave illumination where $U_i(\bm{r}) = A_0$, this further simplifies to:

\begin{equation}
\Gamma_d(\bm{r}_1, \bm{r}_2) = |A_0|^2 \int\int K_{\text{tot}}^*(\bm{r}_1, \bm{r}'_1, z, \lambda) K_{\text{tot}}(\bm{r}_2, \bm{r}'_2, z, \lambda) d^2r'_1 d^2r'_2
\end{equation}

For spatially stationary kernels where $K_S(\bm{r}, \bm{r}', \lambda) = K_S(\bm{r} - \bm{r}', \lambda)$ and appropriate propagation conditions, this reduces to:

\begin{equation}
\Gamma_d(\bm{r}_1, \bm{r}_2) = |A_0|^2 K_{\text{tot}}(\bm{r}_1 - \bm{r}_2, z, \lambda)
\end{equation}

where $K_{\text{tot}}(\bm{r}_1 - \bm{r}_2, z, \lambda)$ is a function that depends only on the separation between the observation points, the propagation distance, and the wavelength.

It is important to emphasize that while this combined formulation provides mathematical and computational convenience, it does not diminish the conceptual value of separating the SECT$_S$ and SECT$_P$ components. The physical insights gained from understanding the distinct roles of surface effects and propagation effects remain a key advantage of our dual-component framework. The combined formulation simply provides an alternative mathematical pathway for implementation that may be advantageous in certain computational scenarios.

Furthermore, the combined formulation allows us to establish direct connections with existing coherence theory, particularly the Van Cittert-Zernike theorem, as demonstrated in Section 2.2. By choosing the appropriate form for the 
surface coherence kernel $K_S$, we can ensure that the total coherence operator $\mathcal{C}_{\text{tot}}$ produces 
results consistent with established coherence theory while maintaining the conceptual and practical advantages of our 
dual-component approach.

In the following sections, we further develop the properties of the combined operator and analyze the relative contributions of the SECT$_S$ and SECT$_P$ components across different optical scenarios, providing practical guidance for applications.

\subsection{Relation to Spatial Filtering Theory} \label{sec:Relation_to_Spatial_Filtering_Theory}

Our framework can be interpreted through the lens of spatial filtering theory, providing additional insight into its operation. For spatially stationary kernels, both the surface coherence operator and the propagation operator act as spatial convolutions:

\begin{equation}
\mathcal{C}_S(U)(\bm{r}) = (K_S * U)(\bm{r})
\end{equation}

\begin{equation}
\mathcal{P}_z(U)(\bm{r}) = (h_z * U)(\bm{r})
\end{equation}

In the spatial frequency domain, these operations correspond to multiplication by the respective Fourier transforms:

\begin{equation}
\mathcal{F}\{\mathcal{C}_S(U)(\bm{r})\}(\bm{\nu}) = \mathcal{F}\{K_S(\bm{r}, \lambda)\}(\bm{\nu}) \cdot \mathcal{F}\{U(\bm{r})\}(\bm{\nu})
\end{equation}

\begin{equation}
\mathcal{F}\{\mathcal{P}_z(U)(\bm{r})\}(\bm{\nu}) = \mathcal{F}\{h(\bm{r}, z, \lambda)\}(\bm{\nu}) \cdot \mathcal{F}\{U(\bm{r})\}(\bm{\nu})
\end{equation}

For the combined operation, we have:

\begin{equation}
\mathcal{F}\{\mathcal{C}_{\text{tot}}(U)(\bm{r})\}(\bm{\nu}) = \mathcal{F}\{K_{\text{tot}}(\bm{r}, z, \lambda)\}(\bm{\nu}) \cdot \mathcal{F}\{U(\bm{r})\}(\bm{\nu})
\end{equation}

where:

\begin{equation}
\mathcal{F}\{K_{\text{tot}}(\bm{r}, z, \lambda)\}(\bm{\nu}) = \mathcal{F}\{h(\bm{r}, z, \lambda)\}(\bm{\nu}) \cdot \mathcal{F}\{K_S(\bm{r}, \lambda)\}(\bm{\nu})
\end{equation}

This perspective reveals how each component modifies the spatial frequency content of the field, with the surface coherence typically attenuating high spatial frequencies (reducing fine spatial structures) and the propagation introducing phase shifts and further filtering based on the propagation distance.

\subsection{Coherence Kernel Construction}

Both the surface coherence kernel and the total coherence kernel must satisfy several physical constraints to generate physically valid partially coherent fields:

\begin{definition}[Admissible Coherence Kernel]
A function $K(\bm{r}, \bm{r}', \lambda)$ is an admissible coherence kernel if:
\begin{enumerate}
\item Hermiticity: $K(\bm{r}, \bm{r}', \lambda) = K^*(\bm{r}', \bm{r}, \lambda)$
\item Positive definiteness: For any function $\psi(\bm{r})$, the following inequality holds:
\begin{equation}
\int \int \psi^*(\bm{r}) K(\bm{r}, \bm{r}', \lambda) \psi(\bm{r}') d^2r d^2r' \geq 0
\end{equation}
\item Normalization: $K(\bm{r}, \bm{r}, \lambda) = 1$ for all $\bm{r}$
\item Spectral scaling: $K(\bm{r}, \bm{r}', \lambda) = f(|\bm{r} - \bm{r}'|/\rho_c(\lambda))$ where $\rho_c(\lambda)$ is the coherence length
\end{enumerate}
\end{definition}

For a Gaussian source, a typical choice for the surface coherence kernel is:
\begin{equation}
K_S(\bm{r}, \bm{r}', \lambda) = \exp\left(-\frac{|\bm{r} - \bm{r}'|^2}{2\rho_c^2(\lambda)}\right)
\end{equation}

where $\rho_c(\lambda) = \lambda/\theta$ is the coherence length, with $\theta$ representing the angular extent of the source.

The propagation kernel, based on the Fresnel approximation, has a quadratic phase dependence:
\begin{equation}
h(\bm{r} - \bm{r}', z, \lambda) = \frac{e^{ikz}}{i\lambda z} e^{i\frac{k}{2z}|\bm{r} - \bm{r}'|^2}
\end{equation}

The combined kernel inherits properties from both components, with the propagation distance $z$ creating a scale-dependent evolution of the coherence properties established at the surface.

\subsection{Spectral and Temporal Generalization}

For polychromatic light, we extend our framework to incorporate spectral effects:
\begin{equation}
\mathcal{C}_{\text{poly}} = \int d\lambda \, W(\lambda) \mathcal{C}_{\text{tot}}(\lambda, z)
\end{equation}

where $W(\lambda)$ is the spectral weight function and $\mathcal{C}_{\text{tot}}(\lambda, z)$ is the total coherence operator for wavelength $\lambda$ and propagation distance $z$.

The framework extends naturally to temporal coherence through:
\begin{equation}
\mathcal{C}_{\text{total}} = \mathcal{C}_{\text{spatial}} \otimes \mathcal{C}_{\text{temporal}}
\end{equation}

This comprehensive framework allows modeling of complex partial coherence effects including both spatial and temporal aspects.

\section{Relative Importance of Surface and Propagation Effects} \label{sec:relative_importance}

A key advantage of our dual-component framework is that it explicitly separates the contributions of surface interactions (SECT$_S$) and propagation effects (SECT$_P$) to the overall coherence transformation. This separation allows us to analyze their relative importance in different optical scenarios, providing valuable insight for practical applications and computational optimization.

\subsection{Theoretical Analysis}

The relative importance of the surface component (SECT$_S$) versus the propagation component (SECT$_P$) depends on several key factors that determine how coherence evolves from the surface to the detection plane:

\begin{enumerate}
\item \textbf{Source coherence characteristics}: The initial coherence properties of the light incident on the surface establish a baseline that can be modified by both components.

\item \textbf{Surface properties}: The microscopic and macroscopic properties of the reflecting surface determine the magnitude of coherence modification at the surface through the SECT$_S$ component.

\item \textbf{Propagation distance}: The distance from the surface to the detection plane determines the extent of coherence evolution during propagation through the SECT$_P$ component.

\item \textbf{Wavelength and characteristic lengths}: The relationship between wavelength, coherence length, surface feature size, and propagation distance determine the scaling of both SECT$_S$ and SECT$_P$ effects.
\end{enumerate}

To quantify the relative importance of these components, we define the following parameters:
\begin{itemize}
\item $\rho_i$ - Initial coherence length of the incident field
\item $\rho_s$ - Coherence length after surface interaction (SECT$_S$)
\item $\rho_d$ - Coherence length at the detection plane (after SECT$_P$)
\item $\sigma_s$ - Characteristic scale of surface features
\item $z$ - Propagation distance
\item $z_R = \rho_s^2/\lambda$ - Rayleigh range for the coherence structure
\end{itemize}

We can define a dimensionless parameter $\eta$ that quantifies the relative contribution of the SECT$_S$ component to the total coherence transformation:

\begin{equation}
\eta = \frac{|\rho_i - \rho_s|}{|\rho_i - \rho_d|}
\end{equation}

This parameter ranges from 0 (surface effects negligible) to 1 (propagation effects negligible). Values near 0.5 indicate comparable contributions from both components.

The surface component (SECT$_S$) typically dominates when:
\begin{itemize}
\item $\sigma_s \approx \lambda$ (surface features comparable to wavelength)
\item $z \ll z_R$ (near-field regime)
\item $\rho_s \ll \rho_i$ (surface significantly reduces coherence)
\end{itemize}

Under these conditions, $\eta$ approaches 1, indicating that the coherence transformation occurs primarily at the surface, with minimal evolution during subsequent propagation.

The propagation component (SECT$_P$) typically dominates when:
\begin{itemize}
\item $\sigma_s \gg \lambda$ (smooth surface)
\item $z \gg z_R$ (far-field regime)
\item $\rho_d \ll \rho_s$ (significant coherence evolution during propagation)
\end{itemize}

In these cases, $\eta$ approaches 0, indicating that the surface interaction preserves most of the incident coherence properties, with the primary coherence transformation occurring during propagation.

We can further characterize these regimes through scaling laws. For the SECT$_S$ component, the coherence modification scales as:

\begin{equation}
\frac{\rho_s}{\rho_i} \approx \begin{cases}
1 & \text{if } \sigma_s \gg \lambda \text{ (smooth surface)} \\
\frac{\sigma_s}{\lambda} & \text{if } \sigma_s \approx \lambda \text{ (moderately rough surface)} \\
\frac{\sigma_s^2}{\lambda^2} & \text{if } \sigma_s \ll \lambda \text{ (very rough surface)}
\end{cases}
\end{equation}

For the SECT$_P$ component, the coherence evolution scales approximately as:

\begin{equation}
\frac{\rho_d}{\rho_s} \approx \begin{cases}
1 & \text{if } z \ll z_R \text{ (near field)} \\
\sqrt{\frac{z}{z_R}} & \text{if } z \approx z_R \text{ (transition region)} \\
\frac{z}{z_R} & \text{if } z \gg z_R \text{ (far field)}
\end{cases}
\end{equation}

These scaling relationships provide a quantitative basis for determining which component will dominate in a given optical 
scenario.

Figure~\ref{fig:Regions_of_Validity_for_Paraxial_Approximation} illustrates the validity regions for the paraxial approximation 
across different optical systems, showing that the ratio of lateral dimensions to propagation distance $L/z$ must typically 
remain below 0.2 to maintain wavefront errors below \(\lambda/20\). Note that curved surfaces further restrict this 
validity region, 
requiring more precise treatment of the SECT$_P$ component through the exact propagation kernel in such cases.

\begin{figure}[H]
    \centering
    \includegraphics[width=0.9\textwidth]{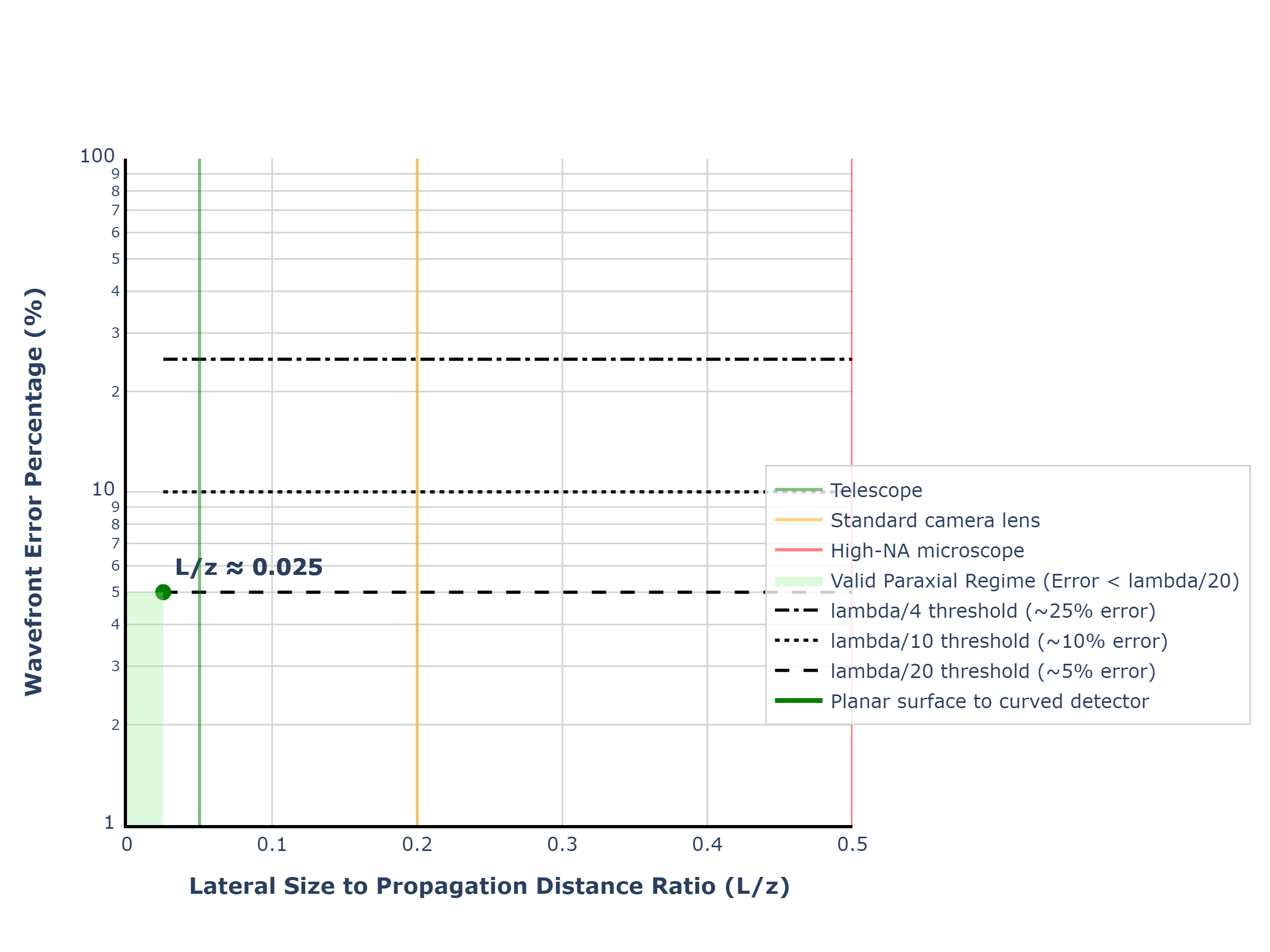}
    \caption{Validity regions for the paraxial approximation. The approximation is valid when the ratio of lateral 
	dimensions \(L\) to propagation distance \(z\) is small, typically requiring $L/z < 0.2$
for phase errors below \(\lambda/20\). Curved surfaces further restrict the valid region.}
    \label{fig:Regions_of_Validity_for_Paraxial_Approximation}
\end{figure}

\subsection{Comparative Analysis Across Optical Scenarios}

Table~\ref{table:relativeImportance} provides a comparative analysis of the relative importance of surface effects (SECT$_S$) and propagation effects (SECT$_P$) in various optical scenarios, offering guidance for practical applications.

\begin{table}[H]
\centering
\caption{Relative importance of surface effects (SECT$_S$) and propagation effects (SECT$_P$) in various optical scenarios. The dominant component indicates which effect should receive greater attention in modeling.}
\label{table:relativeImportance}
\resizebox{\textwidth}{!}{%
\begin{tabular}{|l|P{3.2cm}|P{3cm}|P{3cm}|P{3cm}|}
\hline
\textbf{Scenario} & \textbf{Surface Effects (SECT$_S$)} & \textbf{Propagation Effects (SECT$_P$)} & \textbf{Dominant Component} \\
\hline
Stellar observation through smooth telescope optics & Minor - Surface introduces minimal coherence modification & Major - Source angular size and propagation distance determine coherence & SECT$_P$ \\
\hline
Laser reflection from rough surface & Major - Surface roughness introduces significant coherence reduction & Minor - Coherence pattern mostly determined at surface & SECT$_S$ \\
\hline
Interferometric system with multiple surfaces & Major - Each surface modifies coherence & Moderate - Propagation between surfaces affects fringe visibility & Both \\
\hline
Near-field measurements (z $\ll$ z$_R$) & Major - Surface effects dominate at short distances & Minor - Insufficient propagation for significant coherence evolution & SECT$_S$ \\
\hline
Far-field measurements (z $\gg$ z$_R$) & Moderate - Initial condition for propagation & Major - Extensive coherence evolution during propagation & SECT$_P$ \\
\hline
Spatially varying surfaces (e.g., diffraction gratings) & Major - Surface introduces complex spatial coherence structure & Moderate - Propagation redistributes but preserves coherence characteristics & SECT$_S$ \\
\hline
\end{tabular}
}
\end{table}

For each scenario in Table~\ref{table:relativeImportance}, we can estimate the $\eta$ parameter:

\begin{itemize}
\item Stellar observation: $\eta \approx 0.1$ (SECT$_P$ dominant)
\item Laser reflection from rough surface: $\eta \approx 0.9$ (SECT$_S$ dominant)
\item Interferometric system: $\eta \approx 0.4-0.6$ (both components significant)
\item Near-field measurements: $\eta \approx 0.8$ (SECT$_S$ dominant)
\item Far-field measurements: $\eta \approx 0.2$ (SECT$_P$ dominant)
\item Spatially varying surfaces: $\eta \approx 0.7$ (SECT$_S$ dominant)
\end{itemize}

These estimates provide a quantitative basis for determining which component of our framework requires more detailed 
modeling in different applications.

Figure~\ref{fig:SECT_Component_Application_to_Curved_System} illustrates the application of our dual-component SECT 
framework to a curved optical system, showing how coherence properties are encoded at the reflective surface by 
the SECT$_S$ component and subsequently evolved during propagation to the detector by the SECT$_P$ component. 
This visualization demonstrates the complementary nature of the two components, with coherence length increasing from 
approximately 1-2 units at the surface to 2.5-5.5 units at the detection plane.

\begin{figure}[H]
    \centering
	\includegraphics[trim=1 1 1 1, clip, width=0.9\linewidth]{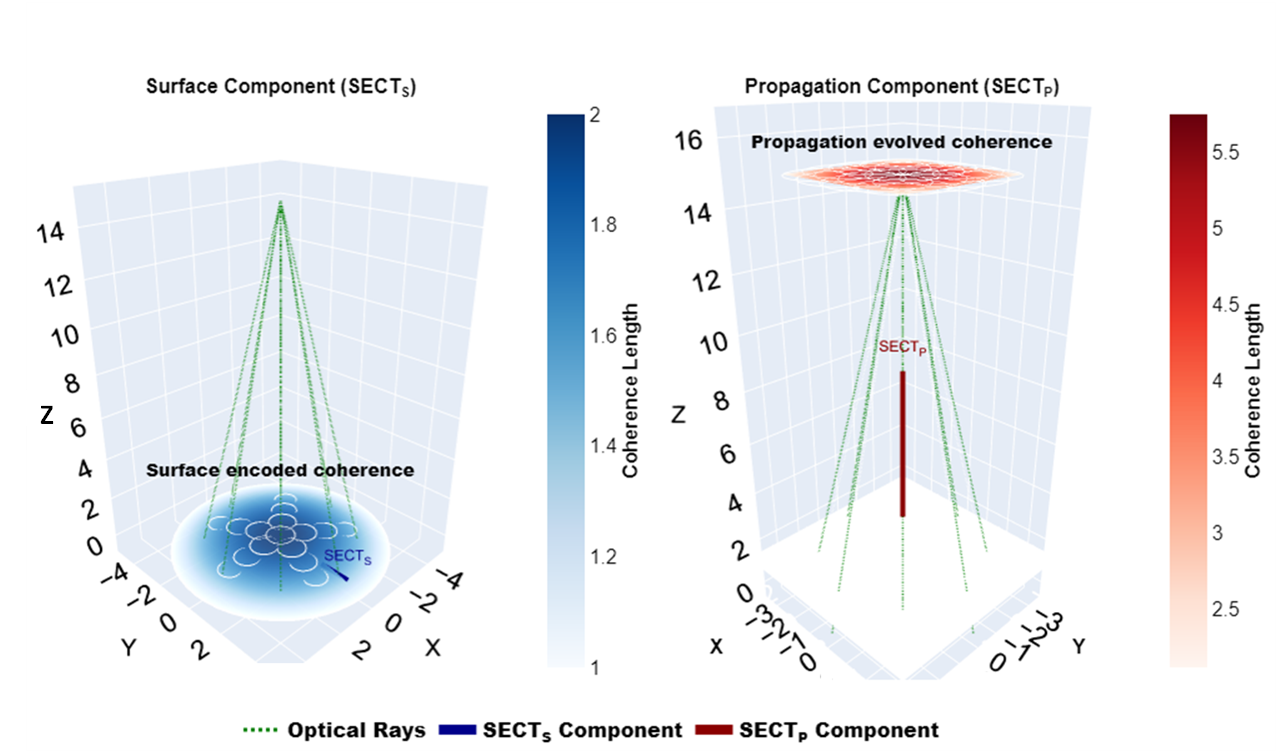}
    \caption{Application of the dual-component SECT framework to a curved optical system. The SECT$_S$
component encodes coherence properties at the surface, while the SECT$_P$
component evolves these properties during propagation to the detector.}
    \label{fig:SECT_Component_Application_to_Curved_System}
\end{figure}

It is worthwhile to note that for a curved primary mirror with radius of curvature $R$ observing a 
distant star, the surface curvature introduces spatially varying path differences across the aperture. 
In this scenario, the SECT$_S$ component must account for the curvature-induced 
phase $\phi_s(\bm{r}) = 2k \cdot |\bm{r}|^2/(2R)$, while the SECT$_P$ component must use the generalized propagation 
kernel that accounts for the actual distance from each mirror point to each detection point. For large 
focal-ratio systems ($f/\# > 5$), the paraxial approximation remains valid, and $\eta \approx 0.3$, 
indicating propagation effects dominate. However, for fast optical systems ($f/\# < 2$), the paraxial 
approximation breaks down, and both components require non-paraxial treatment.

\subsection{Practical Guidelines for Applications}

Based on our theoretical analysis and the comparative assessment across scenarios, we offer the following practical guidelines for applying the SECT framework:

\begin{enumerate}
\item \textbf{For near-field applications}: Focus on accurate modeling of the surface coherence kernel in the SECT$_S$ 
component, as the surface effects dominate. The SECT$_P$ component can often be simplified or linearized without 
significant loss of accuracy. This applies to applications such as near-field microscopy, surface metrology, and 
short-distance laser scanning.

\item \textbf{For far-field applications}: While the SECT$_S$ component provides the initial conditions, the SECT$_P$ 
component requires careful attention as it dominates the final coherence properties. This is particularly relevant 
for astronomical observations, long-distance imaging, and remote sensing applications.
A detailed application of this framework to astronomical imaging with segmented mirror telescopes is provided in 
Appendix~\ref{appendix:Stellar_Coherence_Transformation}.
.

\item \textbf{For rough surface interactions}: The SECT$_S$ component is critical and should be modeled with high 
fidelity, potentially using measured or empirically derived coherence kernels rather than analytical approximations. 
Applications include laser speckle imaging, optical coherence tomography of rough tissues, and scattering-based material 
characterization.

\item \textbf{For multi-surface systems}: Both SECT$_S$ and SECT$_P$ components are important, and their sequential 
application must be carefully tracked through the system. The interplay between surface effects and propagation 
effects can lead to complex coherence evolution that requires the full dual-component framework. 
This applies to complex optical systems such as multi-element imaging systems, interferometers, and 
cascaded diffractive optics.

\item \textbf{For wavelength-dependent systems}: The scaling of both SECT$_S$ and SECT$_P$ components with wavelength 
must be accounted for, especially in broadband applications where the relative importance can vary across the spectrum. 
This is crucial for spectroscopic systems, chromatic aberration analysis, and polychromatic imaging.
\end{enumerate}

These guidelines may help practitioners identify which component requires more detailed modeling based on their specific application, enabling more efficient and accurate simulation of partial coherence effects. The separation of surface and propagation effects in our dual-component framework provides not only conceptual clarity but also practical computational advantages by allowing selective optimization of each component.

For computational implementation, this separation allows:

\begin{itemize}
\item Parallel processing of different surface interactions in multi-surface systems
\item Selective refinement of the dominant component for a given scenario
\item Pre-computation of surface coherence kernels for repeated use in parametric studies
\item Adaptive mesh refinement based on the relative importance of each component
\end{itemize}

The ability to tailor computational resources based on the relative importance of each component represents a significant advantage of our dual-component approach over traditional monolithic coherence models.

\section{Mathematical Properties} \label{sec:mathematical_properties}

\subsection{Convergence and Equivalence}

Both components of our framework, as well as their combination, satisfy important convergence properties that ensure their physical validity.

\begin{theorem}[Convergence to Ensemble Average] \label{thm:convergence}
Let $A$ denote the effective area over which the incident field is defined, and let $\rho_c$ be the transverse coherence radius of the source. Define the number of spatial coherence cells as $N_c \approx A / (\pi \rho_c^2)$. In the limit where $\rho_c \to 0$ and $N_c \to \infty$ such that the total optical power remains finite, the combined surface-encoded and propagation coherence transformation converges to the traditional ensemble average:
\begin{equation}
\lim_{N_c \to \infty} \langle |U_d|^2 \rangle_{\text{SECT}} = \langle |U_d|^2 \rangle_{\text{ensemble}}
\end{equation}
where $\langle |U_d|^2 \rangle_{\text{SECT}}$ denotes the intensity calculated using the SECT framework, and $\langle |U_d|^2 \rangle_{\text{ensemble}}$ denotes the intensity calculated by traditional ensemble averaging over multiple statistically independent realizations.
\end{theorem}

The condition $\rho_c \to 0$ ensures spatial incoherence across cells, aligning the SECT-imposed surface coherence structure with the delta-correlated statistics assumed in ensemble averaging. The relationship $N_c \approx A / (\pi \rho_c^2)$ quantifies the number of statistically independent regions across the surface. While the theorem establishes mean-square convergence, this work does not derive an explicit convergence rate.

\begin{proof}
To establish this convergence, we must first precisely define what we mean by "coherence cells" and establish their relationship to the spatial sampling of our operators.

\textbf{Definition of Coherence Cells:} A coherence cell is defined as a region over which the field maintains a high degree of spatial correlation. Quantitatively, for a partially coherent field with coherence function $\gamma(\bm{r}_1, \bm{r}_2)$, a coherence cell centered at $\bm{r}_0$ is the region $\Omega_{\bm{r}_0}$ such that for all $\bm{r} \in \Omega_{\bm{r}_0}$, $|\gamma(\bm{r}_0, \bm{r})| \geq 1/e$.

For a field with coherence length $\rho_c$, the area of a coherence cell is approximately $A_{\text{cell}} \approx \pi\rho_c^2$. Given a computational domain of area $A$, the number of coherence cells is:
\begin{equation}
N_c \approx \frac{A}{\pi\rho_c^2}
\end{equation}

In our SECT framework, the intensity at the detection plane is given by:
\begin{equation}
\langle |U_d(\bm{r})|^2 \rangle_{\text{SECT}} = \langle |\mathcal{P}_z(\mathcal{C}_S(\mathcal{R}(U_i)))(\bm{r})|^2 \rangle
\end{equation}

Let us expand this using the integral representation of our operators:
\begin{align}
\langle |U_d(\bm{r})|^2 \rangle_{\text{SECT}} 
&= \left\langle \left| \int h(\bm{r}, \bm{r}'', z, \lambda) \right. \right. \\
&\quad \left. \left. \times \left( \int K_S(\bm{r}'', \bm{r}', \lambda) \notag
\mathcal{R}(U_i)(\bm{r}') \, d^2r' \right) d^2r'' \right|^2 \right\rangle
\end{align}

In the traditional ensemble averaging approach, we generate multiple realizations of the field $U_i^{(j)}(\bm{r})$ ($j = 1, 2, \ldots, N_{\text{realizations}}$) with appropriate statistical properties, and then compute:
\begin{equation}
\begin{aligned}
\langle |U_d(\bm{r})|^2 \rangle_{\text{ensemble}} 
&= \lim_{N_{\text{realizations}} \to \infty} \frac{1}{N_{\text{realizations}}} 
\sum_{j=1}^{N_{\text{realizations}}} |\mathcal{P}_z(\mathcal{R}(U_i^{(j)}))(\bm{r})|^2 \\
&= \lim_{N_{\text{realizations}} \to \infty} \frac{1}{N_{\text{realizations}}} 
\sum_{j=1}^{N_{\text{realizations}}} \left| 
\int h(\bm{r}, \bm{r}', z, \lambda) \right. \\
&\quad \left. \times\, \mathcal{R}(U_i^{(j)})(\bm{r}')\, d^2r' \right|^2
\end{aligned}
\end{equation}

To establish the convergence, we need to show that:
\begin{equation}
\lim_{N_c \to \infty} \langle |U_d(\bm{r})|^2 \rangle_{\text{SECT}} = \lim_{N_{\text{realizations}} \to \infty} \langle |U_d(\bm{r})|^2 \rangle_{\text{ensemble}}
\end{equation}

The key insight is that as $N_c \to \infty$, the coherence length $\rho_c \to 0$, which means the surface coherence 
kernel $K_S$ approaches a delta function. This corresponds to a field with increasingly fine-scale fluctuations, 
effectively sampling different statistical realizations within the computational domain.

We can formalize this by discretizing our computational domain into a grid with spacing $\Delta x, \Delta y \ll \rho_c$, 
giving us a total of $N_{\text{grid}} = \frac{A}{\Delta x \Delta y}$ sampling points. For sufficiently fine sampling, 
each coherence cell contains approximately $N_{\text{points/cell}} = \frac{\pi\rho_c^2}{\Delta x \Delta y}$ grid points.

From the central limit theorem, when we apply our surface coherence operator $\mathcal{C}_S$, the field at each grid 
point represents a weighted average of $N_{\text{points/cell}}$ statistically related points. As $N_c \to \infty$ and 
consequently $\rho_c \to 0$, each coherence cell becomes smaller, and $N_{\text{points/cell}} \to 1$. At this limit, 
the coherence operator effectively transforms each grid point independently according to the statistical properties 
encoded in $K_S$.

To see how this relates to ensemble averaging, consider a discretized version of our computational domain 
with $N_{\text{grid}}$ points. In the SECT approach, as $N_c \to \infty$ (and $\rho_c \to 0$), each grid point 
effectively represents an independent random sample from the ensemble of possible field values. The propagation 
operator $\mathcal{P}_z$ then combines these independent samples with appropriate weights.

In the limit of $N_c \to \infty$, the SECT approach effectively samples $N_{\text{grid}}$ independent realizations 
of the field within a single computation, which mathematically converges to the ensemble average over multiple 
realizations as $N_{\text{grid}} \to \infty$.

Formally, for a discretized grid, we can represent the SECT approach as:
\begin{align}
\langle |\bm{U}_d(\bm{r})|^2 \rangle_{\text{SECT}} 
&= \left\langle \left| \sum_{m,n} h(\bm{r}, \bm{r}_{m,n}, z, \lambda)\, \Delta x\, \Delta y \right. \right. \notag \\
&\quad \left. \left. \times \left( \sum_{i,j} K_S(\bm{r}_{m,n}, \bm{r}_{i,j}, \lambda)\, \mathcal{R}(\bm{U}_i)(\bm{r}_{i,j})\, \Delta x\, \Delta y \right) \right|^2 \right\rangle
\end{align}

As $N_c \to \infty$ and $\rho_c \to 0$, the kernel $K_S$ approaches a delta function, and each term in the inner sum represents an independent sample from the ensemble. 

Therefore, by the law of large numbers, as the number of coherence cells approaches infinity:
\begin{equation}
\lim_{N_c \to \infty} \langle |U_d(\bm{r})|^2 \rangle_{\text{SECT}} = \langle |U_d(\bm{r})|^2 \rangle_{\text{ensemble}}
\end{equation}

This establishes the convergence of our SECT framework to the traditional ensemble average in the limit of an infinite number of coherence cells.
\end{proof}

\subsection{Energy Conservation}

\begin{proposition}[Energy Conservation] \label{prop:energy}
The combined coherence operator preserves energy in the sense:
\begin{equation}
\int |\mathcal{C}_{\text{tot}}(U)(\bm{r})|^2 d^2r = \int |U(\bm{r})|^2 d^2r
\end{equation}
\end{proposition}

This property ensures that our framework does not artificially amplify or attenuate the total field energy, maintaining physical consistency.

\subsection{Operational Properties}

Both the surface coherence operator and the combined operator exhibit several useful mathematical properties:

\begin{lemma}[Operator Properties] \label{lem:properties}
The coherence operators satisfy:
\begin{enumerate}
\item Linearity: For any fields $U_1$ and $U_2$ and scalars $a$ and $b$:
\begin{equation}
\mathcal{C}(aU_1 + bU_2) = a\mathcal{C}(U_1) + b\mathcal{C}(U_2)
\end{equation}

\item Conservation of average intensity: For a spatially uniform field $U$, the average intensity is preserved:
\begin{equation}
\frac{1}{A}\int_A |\mathcal{C}(U)(\bm{r})|^2 d^2r = \frac{1}{A}\int_A |U(\bm{r})|^2 d^2r
\end{equation}
where $A$ is the area of interest.

\item Spectral bounds: When viewed as integral operators, all eigenvalues $\lambda_i$ satisfy:
\begin{equation}
0 \leq \lambda_i \leq 1
\end{equation}
which ensures that the operators do not artificially amplify any field components.
\end{enumerate}
\end{lemma}

These properties establish that our framework behaves as a physically reasonable transformation that preserves the essential characteristics of the optical field while introducing and evolving partial coherence effects.

The mathematical properties established in the preceding sections—convergence, energy conservation, and operational characteristics—provide the theoretical foundation for translating our framework into computational representations. This transition from mathematical formalism to computational structure is not merely an implementation detail, but rather a direct consequence of the framework's mathematical architecture.

The convergence properties established in Theorem~\ref{thm:convergence} inform computational discretization by providing theoretical bounds on sampling requirements. As we demonstrated, for a sufficiently large number of coherence cells $N_c$, our SECT framework converges to the traditional ensemble average. This convergence theorem provides a theoretical basis for determining appropriate spatial sampling in computational implementations, where each coherence cell must be adequately represented to ensure physical accuracy.

Similarly, the energy conservation property (Proposition~\ref{prop:energy}) provides a critical validation mechanism for computational implementations. The requirement that $\int |\mathcal{C}_{\text{tot}}(U)(\bm{r})|^2 d^2r = \int |U(\bm{r})|^2 d^2r$ translates to a discrete conservation constraint that any valid computational representation must satisfy. This property serves as both a theoretical constraint on valid discretization schemes and a practical tool for verifying computational accuracy.

The linearity and spectral bounding properties of our operators (Lemma~\ref{lem:properties}) have particularly significant computational implications. Linearity enables the use of superposition principles in computational implementations, allowing complex fields to be decomposed into simpler components that can be processed independently and then recombined. The spectral bounds, which ensure that all eigenvalues satisfy $0 \leq \lambda_i \leq 1$, guarantee numerical stability in iterative computational approaches by preventing unbounded amplification of field components.

These mathematical properties not only validate our framework theoretically but directly shape its computational representation. The dual-component structure of the SECT framework, with its separable surface and propagation operators, manifests in the computational domain as a natural decomposition that aligns with the underlying physics while enabling significant reductions in computational complexity. In the following section, we explore these computational implications without focusing on implementation specifics, instead examining how the mathematical structure of our framework fundamentally transforms the computational approach to coherence modeling.

\section{Theoretical Implications for Computational Implementation}

A notable theoretical consequence of our dual-component SECT framework is its potential to transform the computational 
approach to modeling partially coherent optical systems. 
This section examines the mathematical structure of our framework from a computational perspective, analyzing 
how the formal separation of surface and propagation effects translates into fundamental changes in 
computational representation and complexity. 
Rather than providing specific implementation details, we focus on how the theoretical structure informs computational 
approaches.
The mathematical separation of coherence effects into surface-encoded and propagation components provides not 
only conceptual clarity but also reveals a natural computational structure that mirrors the physics 
of coherence evolution. By examining how this mathematical structure maps to discretized representations, we 
can analyze the theoretical computational advantages that emerge from our formalism without delving into specific 
numerical algorithms or optimization techniques.
Our analysis centers on three key theoretical aspects with computational implications: (1) the natural discretization 
of our mathematical operators, (2) the computational advantages inherent in the dual-component structure, 
and (3) the fundamental reduction in computational complexity that emerges for multi-surface systems. 
This theoretical analysis demonstrates how the mathematical properties established in previous sections 
translate to computational representations, while maintaining the paper's focus on the formal mathematical 
framework rather than implementation specifics.

\subsection{Discrete Formulation}

To implement the SECT framework numerically, we must discretize both the SECT$_S$ and SECT$_P$ components. For a computational grid with spacing $\Delta x$ and $\Delta y$, the discretized forms of these components are developed below.

\subsubsection{Discretization of the SECT$_S$ Component}

The surface coherence operator $\mathcal{C}_S$ is discretized as a matrix operator:
\begin{equation}
[\mathcal{C}_S]_{mn,ij} = K_S(m\Delta x, n\Delta y, i\Delta x, j\Delta y, \lambda) \Delta x \Delta y
\end{equation}

where $(m,n)$ and $(i,j)$ are grid indices for output and input points, respectively. The area element $\Delta x \Delta y$ ensures proper normalization of the discrete integral.

For spatially stationary kernels where $K_S$ depends only on the difference between coordinates, this simplifies to:
\begin{equation}
[\mathcal{C}_S]_{mn,ij} = K_S((m-i)\Delta x, (n-j)\Delta y, \lambda) \Delta x \Delta y
\end{equation}

When applied to a discrete field $U_{ij}$, the operation becomes:
\begin{equation}
[U_s]_{mn} = \sum_{i,j} [\mathcal{C}_S]_{mn,ij} [U_i]_{ij}
\end{equation}

For large grids, this direct matrix multiplication would be computationally prohibitive, scaling as $O(N^4)$ for an $N \times N$ grid. However, for spatially stationary kernels, we can leverage the convolution theorem and implement this operation efficiently using Fast Fourier Transforms (FFTs), reducing the computational complexity to $O(N^2 \log N)$.

\subsubsection{Discretization of the SECT$_P$ Component}

The propagation operator $\mathcal{P}_z$ is discretized similarly:
\begin{equation}
[\mathcal{P}_z]_{mn,ij} = h(m\Delta x, n\Delta y, i\Delta x, j\Delta y, z, \lambda) \Delta x \Delta y
\end{equation}

For the Fresnel propagation kernel:
\begin{equation}
[\mathcal{P}_z]_{mn,ij} = \frac{e^{ikz}}{i\lambda z} e^{i\frac{k}{2z}[(m\Delta x - i\Delta x)^2 + (n\Delta y - j\Delta y)^2]} \Delta x \Delta y
\end{equation}

This can also be implemented efficiently using the angular spectrum method or the Fresnel transfer function approach, both of which utilize FFTs to achieve $O(N^2 \log N)$ complexity.

\subsubsection{Combined Implementation}

When implementing the complete SECT framework, we apply the SECT$_S$ and SECT$_P$ components sequentially:
\begin{equation}
[U_d] = [\mathcal{P}_z][\mathcal{C}_S][\mathcal{R}][U_i]
\end{equation}

For certain applications, it may be more efficient to combine the operators into a single operation:
\begin{equation}
[U_d] = [\mathcal{C}_{\text{tot}}][\mathcal{R}][U_i]
\end{equation}

where $[\mathcal{C}_{\text{tot}}] = [\mathcal{P}_z][\mathcal{C}_S]$ represents the combined effect of surface coherence 
transformation and propagation.

For curved surfaces and non-planar detection geometries, the discretization must account for the actual three-dimensional coordinates of each point. The propagation kernel becomes:
\begin{equation}
[\mathcal{P}_z]_{mn,ij} = \frac{1}{i\lambda |\bm{r}_{mn} - \bm{r}_{ij}|} e^{ik|\bm{r}_{mn} - \bm{r}_{ij}|} \Delta A_{ij}
\end{equation}
where $\bm{r}_{mn}$ represents the three-dimensional coordinates of point $(m,n)$ on the detection surface, $\bm{r}_{ij}$ represents the coordinates of point $(i,j)$ on the reflecting surface, and $\Delta A_{ij}$ is the differential area element at point $(i,j)$ which may vary across a curved surface. This generalization preserves the mathematical structure of the SECT framework while accounting for arbitrary surface and detector geometries, though at the cost of increased computational complexity as the operation can no longer be implemented using simple FFT-based methods.

\subsection{Computational Advantages of the Dual-Component Approach}

The explicit separation of surface and propagation effects in our SECT framework offers several computational advantages:

\begin{enumerate}
\item \textbf{Selective refinement}: Computational resources can be allocated based on the relative importance of each component. For near-field applications where SECT$_S$ dominates, more computational effort can be devoted to accurately modeling the surface coherence kernel, while using simplified propagation models.

\item \textbf{Component-specific optimization}: Different numerical techniques can be applied to each component based on their mathematical properties. For example, adaptive sampling can be used for the SECT$_S$ component in regions of rapidly varying surface properties.

\item \textbf{Parallel implementation}: The sequential nature of the SECT framework allows for straightforward parallelization, with the SECT$_S$ and SECT$_P$ components potentially distributed across different computational resources.

\item \textbf{Pre-computation}: For systems with fixed surface properties, the SECT$_S$ component can be pre-computed and reused for multiple propagation scenarios, significantly reducing overall computation time.

\item \textbf{Scalability}: The framework scales efficiently to multi-surface systems, with each surface-propagation pair handled sequentially without requiring the simultaneous evaluation of complex multi-surface integrals.
\end{enumerate}

These advantages make the SECT framework particularly well-suited for computationally intensive applications such as multi-surface optical systems, wide-field high-resolution simulations, and parametric studies where multiple configurations need to be evaluated efficiently.

\subsection{Multi-Surface Optical System}

To illustrate the computational efficiency of our framework, consider a multi-surface optical system with $M$ surfaces and a propagation distance $z_i$ between consecutive surfaces. The traditional approach would require solving a $2M$-dimensional integral to capture all coherence interactions, resulting in computational complexity of $O(N^{2M})$ for an $N \times N$ grid.

In contrast, our SECT framework decomposes this into a sequence of SECT$_S$ and SECT$_P$ operations:
\begin{equation}
U_{\text{final}} = \mathcal{P}_{z_M} \circ \mathcal{C}_{S_M} \circ \mathcal{P}_{z_{M-1}} \circ \cdots \circ \mathcal{C}_{S_2} \circ \mathcal{P}_{z_1} \circ \mathcal{C}_{S_1}(U_{\text{initial}})
\end{equation}

Using FFT-based implementations, this reduces the computational complexity to $O(M \cdot N^2 \log N)$, representing an exponential improvement in efficiency that becomes increasingly significant as the number of surfaces grows.

This computational advantage, combined with the physical insights gained from separating surface and propagation effects, makes the SECT framework a powerful tool for modeling partial coherence in complex optical systems.

\subsection{Computational Complexity Analysis and Theoretical Implications}

We now analyze the computational complexity of the SECT framework from a theoretical perspective, connecting 
the mathematical formalism to its computational consequences and relating these findings to the relative 
importance of surface and propagation effects established in Section~\ref{sec:relative_importance}.

In conventional approaches to partial coherence modeling, the mutual coherence function $\Gamma(\bm{r}_1, \bm{r}_2)$ must 
be propagated through each optical surface using integral transforms that scale as $O(N^2)$ per variable, where $N$ is 
the number of spatial sampling points per dimension. For a system with $M$ successive optical surfaces, the composite 
complexity scales as $O(N^{2M})$, assuming no simplifying assumptions on coherence structure or separability. 
This exponential scaling represents a fundamental theoretical limitation of direct mutual coherence propagation 
approaches.

The mathematical separation at the heart of our SECT framework fundamentally transforms this computational structure. By decomposing the problem into $M$ modular stages, each consisting of a surface encoding step and a propagation step, we exploit the mathematical properties established in Section~\ref{sec:mathematical_properties} to achieve dramatically different scaling behavior. The linearity property (Lemma~\ref{lem:properties}) enables the use of Fourier-based methods for both components, while the separability of operators allows each step to be treated independently.

Each surface encoding step via an operator kernel application scales as $O(N^2)$, while each propagation step using FFT-based convolution scales as $O(N^2 \log N)$. Thus, the total computational complexity of SECT scales as:
\begin{equation}
O\left(M \cdot N^2 \log N \right),
\end{equation}
representing an exponential improvement in complexity compared to conventional approaches.

This complexity reduction directly relates to the relative importance analysis in Section~\ref{sec:relative_importance}. For scenarios where the surface component (SECT$_S$) dominates (high $\eta$ values), computational resources can be theoretically concentrated on the surface encoding step, potentially using simplified propagation methods. Conversely, for scenarios where the propagation component (SECT$_P$) dominates (low $\eta$ values), resources can be allocated primarily to accurate propagation computation.

The parameter $\eta$ introduced in Section~\ref{sec:relative_importance} thus serves not only as a physical metric for understanding coherence evolution but also as a theoretical guide for computational optimization. For instance, in near-field applications with $\eta \approx 0.8$ (Table~\ref{table:relativeImportance}), the complexity can be further reduced by using simplified propagation models, whereas in far-field applications with $\eta \approx 0.2$, computational resources should be focused on the propagation component.

The mathematical structure of our framework also enables more sophisticated computational approaches that transcend the basic complexity analysis presented here. The positive-semidefiniteness and spectral bounds of our operators (established in Section~\ref{sec:mathematical_properties}) guarantee that low-rank approximations can be employed when appropriate, potentially further reducing computational requirements. Additionally, the convergence theorem (Theorem~\ref{thm:convergence}) provides a theoretical foundation for adaptive sampling strategies that allocate computational resources based on the local coherence structure.

This complexity analysis illustrates how the mathematical properties of our framework directly translate to computational advantages, demonstrating that the SECT approach represents not just a conceptual reframing of coherence theory but a fundamentally different computational paradigm with significant theoretical advantages for complex optical systems.

\subsection{Limitations and Further Considerations}

While this section has examined the theoretical computational implications of our framework, we acknowledge that 
practical implementation would require addressing several additional considerations beyond the scope of this 
theoretically-focused paper. Numerical stability represents a significant concern, particularly for systems with 
high dynamic range or when propagating over large distances where phase unwrapping and sampling issues may arise. 
The surface-encoded coherence kernel may require specialized discretization schemes to maintain its mathematical 
properties, especially the positive-definiteness constraint established in Section~\ref{sec:mathematical_properties}. 
Algorithm optimization, including techniques for sparse representation of coherence operators and parallel implementation 
strategies, would be essential for handling large-scale systems efficiently. Additionally, error propagation analysis 
would be necessary to establish confidence intervals on the computed coherence properties. These practical considerations, 
while important for implementation, are secondary to the fundamental mathematical structure presented in this paper. 
Our focus has been on establishing the theoretical foundation and computational implications of the dual-component 
framework, providing the mathematical groundwork upon which future implementation efforts can build. We believe this 
theoretical analysis sufficiently demonstrates the potential computational advantages of our approach while 
acknowledging that bridging to practical implementation would require additional engineering considerations.

\section{Conclusions}

In this paper, we have introduced a new framework for modeling partial coherence effects in optical systems through a dual-component approach: the Surface-Encoded Coherence Transformation (SECT) consisting of a surface component (SECT$_S$) and a propagation component (SECT$_P$). This framework provides a complete physical description of how coherence evolves from a reflecting surface to a detection plane, while maintaining mathematical elegance and computational efficiency.

\vspace{0.3cm}

The key contributions of this work can be summarized as follows:

\begin{enumerate}
   \item \textbf{Unified Theoretical Framework}: We have developed a rigorous mathematical formalism that explicitly separates and then recombines surface coherence effects (SECT$_S$) and propagation effects (SECT$_P$). The proof of equivalence with the Van Cittert-Zernike theorem demonstrates that our approach preserves physical accuracy while offering conceptual and computational advantages.
   
   \item \textbf{Conceptual Clarity}: By separating surface and propagation effects into distinct components, our framework provides deeper insight into the physical mechanisms that influence coherence in optical systems. This separation allows for more intuitive understanding of complex coherence phenomena, particularly in systems with multiple surfaces and propagation paths.
   
   \item \textbf{Relative Importance Analysis}: We have provided a quantitative analysis of when surface effects (SECT$_S$) dominate versus when propagation effects (SECT$_P$) dominate, offering practical guidelines for applications across different optical scenarios. The introduction of the parameter $\eta$ provides a rigorous metric for determining which component requires more detailed modeling in specific applications.
   
   \item \textbf{Mathematical Properties}: We have established the key mathematical properties of both components and their combination, including linearity, energy conservation, and spectral bounds. These properties ensure that our framework produces physically valid results consistent with wave optics principles while offering new mathematical insights into coherence evolution.
   
   \item \textbf{Computational Efficiency}: The dual-component structure allows for targeted computational optimizations based on the relative importance of each component in specific scenarios. For multi-surface systems, 
   our approach offers an exponential improvement in computational complexity compared to traditional methods, 
   scaling as $O(M \cdot N^2 \log N)$ rather than $O(N^{2M})$ for a system with $M$ surfaces.
\end{enumerate}

The SECT framework represents a paradigm shift in how we conceptualize and model partial coherence in optical systems. Unlike traditional approaches that treat coherence primarily as a source property or as a propagation effect, our dual-component approach recognizes that coherence evolution involves two distinct physical processes: modification at surfaces (SECT$_S$) and evolution during propagation (SECT$_P$). This perspective not only aligns more closely with the physical reality of how coherence develops in complex optical systems but also offers practical advantages for analysis, design, and computational implementation.

From a practical perspective, the SECT framework offers significant advantages for optical system design and analysis. 
By understanding which component dominates in different scenarios, optimization efforts may be focused on the most 
influential aspects of the systems. For example, in near-field applications where SECT$_S$ dominates, surface 
treatments and materials can be optimized to achieve desired coherence properties without extensive modeling of 
propagation effects. Conversely, in far-field applications where SECT$_P$ dominates, system geometry and propagation 
paths can be optimized while using simplified surface models.

An important distinguishing feature of the proposed SECT framework lies in its deterministic encoding of coherence, in 
contrast to conventional ensemble-based or statistical sampling approaches. By representing partial coherence as an 
explicit surface transformation operator, the framework provides a structured and reproducible means of modeling 
coherence effects without invoking stochastic averaging or mode decomposition. This not only improves computational 
efficiency but also offers a more interpretable link between physical source properties and wavefield evolution. 
The operator-based nature of the SECT formalism further enables integration into analytical propagation schemes, 
making it particularly suitable for high-precision simulation scenarios in structured or engineered optical systems.

A limitation of the current framework is its reliance on the paraxial approximation, which restricts its applicability 
to systems with relatively small angular extents. Extending the SECT approach to incorporate non-paraxial propagation 
would significantly broaden its applicability to high-numerical-aperture systems and extreme off-axis configurations.

In conclusion, our dual-component SECT framework represents both a theoretical advance in our understanding of partial coherence and a practical approach to incorporating coherence effects in wave optics simulations. By explicitly modeling both surface transformations (SECT$_S$) and subsequent propagation (SECT$_P$), we provide a more complete and physically intuitive description of coherence evolution in optical systems, with the potential to significantly impact how we analyze and design systems where coherence effects play a crucial role. This framework bridges the gap between physical understanding and computational implementation, offering a new perspective on one of the fundamental aspects of optical physics.

\newpage

\appendix

\section{Detailed Proof of SECT-VCZ Equivalence}\label{sup:SECT-VCZ}

In this appendix, we provide a complete, detailed derivation establishing the equivalence between our Surface-Encoded 
Coherence Transformation framework and the classical van Cittert-Zernike theorem (see also in 
Section~\ref{sec:brief_outline_of_the_equivalence_proof}. 
This includes demonstrating the 
positive-semidefiniteness of the surface kernel, explicitly including the Fresnel propagator, and showing all intermediate 
steps in applying the convolution theorem.

\subsection{Positive-Semidefiniteness of the Surface Kernel}

First, we demonstrate that the surface coherence kernel $K_S$ is positive-semidefinite, a necessary condition for it to 
represent a physically valid coherence function.

\begin{lemma}[Positive-Semidefiniteness of Surface Kernel]
The surface coherence kernel 
$K_S(\mathbf{r}, \mathbf{r}', \lambda) = \mathcal{F}\left\{I_s\left(\frac{\mathbf{\rho}}{\lambda z}\right)\right\}(\mathbf{r} - \mathbf{r}')$ 
is positive-semidefinite.
\end{lemma}

\begin{proof}
For any function $\psi(\mathbf{r})$ with compact support, we need to show that:
\begin{equation}
\iint \psi^*(\mathbf{r}) K_S(\mathbf{r}, \mathbf{r}', \lambda) \psi(\mathbf{r}') \, d^2r \, d^2r' \geq 0
\end{equation}

Substituting the explicit form of $K_S$:
\begin{equation}
\iint \psi^*(\mathbf{r}) \mathcal{F}\left\{I_s\left(\frac{\mathbf{\rho}}{\lambda z}\right)\right\}(\mathbf{r} - \mathbf{r}') \psi(\mathbf{r}') \, d^2r \, d^2r'
\end{equation}

Let us define the function $g(\mathbf{R}) = \mathcal{F}\left\{I_s\left(\frac{\mathbf{\rho}}{\lambda z}\right)\right\}(\mathbf{R})$. By the convolution theorem, the above integral can be rewritten as:
\begin{equation}
\iint \psi^*(\mathbf{r}) g(\mathbf{r} - \mathbf{r}') \psi(\mathbf{r}') \, d^2r \, d^2r' = \int \left|\mathcal{F}\{\psi\}(\mathbf{k})\right|^2 \cdot \tilde{g}(\mathbf{k}) \, d^2k
\end{equation}

where $\tilde{g}(\mathbf{k}) = I_s\left(\frac{\mathbf{k}}{\lambda z}\right)$ is the Fourier transform of $g$. Since $I_s$ represents an intensity distribution, $I_s \geq 0$ for all $\mathbf{\rho}$. Therefore, $\tilde{g}(\mathbf{k}) \geq 0$ for all $\mathbf{k}$, and consequently:
\begin{equation}
\int \left|\mathcal{F}\{\psi\}(\mathbf{k})\right|^2 \cdot \tilde{g}(\mathbf{k}) \, d^2k \geq 0
\end{equation}

This proves that the kernel $K_S$ is positive-semidefinite.
\end{proof}

\subsection{Step-by-Step Derivation of SECT-VCZ Equivalence}

Now we proceed with the complete derivation establishing the equivalence between our framework and the classical VCZ theorem.

\begin{theorem}[SECT-VCZ Equivalence with Detailed Steps]
Given a fully coherent incident field $U_i(\mathbf{r}) = A_0$ (uniform plane wave illumination) impinging on a 
surface with coherence operator $\mathcal{C}_S$ characterized by kernel $K_S(\mathbf{r}, \mathbf{r}', \lambda)$, 
followed by propagation operator $\mathcal{P}_z$, the mutual coherence function at the detection plane is mathematically 
equivalent to that obtained from the VCZ theorem for an incoherent source with intensity distribution $I_s(\mathbf{\rho})$, 
provided that:

\begin{equation}
K_S(\mathbf{r}, \mathbf{r}', \lambda) = \mathcal{F}\left\{I_s\left(\frac{\mathbf{\rho}}{\lambda z}\right)\right\}(\mathbf{r} - \mathbf{r}')
\end{equation}
where $\mathcal{F}$ denotes the Fourier transform operator.
\end{theorem}

\begin{proof}
We start with the mutual coherence function at the detection plane using our SECT framework:
\begin{align}
\Gamma_{\text{SECT}}(\mathbf{r}_1, \mathbf{r}_2) &= \langle U_d^*(\mathbf{r}_1) U_d(\mathbf{r}_2) \rangle
\end{align}

Where $U_d(\mathbf{r})$ is the field at the detection plane after applying the reflection operator $\mathcal{R}$, the surface coherence operator $\mathcal{C}_S$, and the propagation operator $\mathcal{P}_z$. We can express this as:
\begin{align}
U_d(\mathbf{r}) &= \mathcal{P}_z(\mathcal{C}_S(\mathcal{R}(U_i)))(\mathbf{r})
\end{align}

With the explicit Fresnel propagator, this becomes:
\begin{align}
U_d(\mathbf{r}) &= \int h(\mathbf{r}, \mathbf{r}'', z, \lambda) \left( \int K_S(\mathbf{r}'', \mathbf{r}', \lambda) \mathcal{R}(U_i)(\mathbf{r}') \, d^2r' \right) \, d^2r'' \\
&= \int \frac{e^{ikz}}{i\lambda z} e^{i\frac{k}{2z}|\mathbf{r} - \mathbf{r}''|^2} \left( \int K_S(\mathbf{r}'', \mathbf{r}', \lambda) \mathcal{R}(U_i)(\mathbf{r}') \, d^2r' \right) \, d^2r'' \notag
\end{align}
where $k = 2\pi/\lambda$ is the wavenumber.

For simplicity, we assume a deterministic reflection operator $\mathcal{R}$ with uniform reflection coefficient $R$ and phase change $\phi_R$:
\begin{align}
\mathcal{R}(U_i)(\mathbf{r}) = R e^{i\phi_R} U_i(\mathbf{r}) = R e^{i\phi_R} A_0 \equiv B_0
\end{align}
where we define $B_0 = R e^{i\phi_R} A_0$ for notational simplicity.

To accurately model the statistical properties of partial coherence, we recognize that the surface coherence operator introduces statistical variations while preserving the mean field. The key insight is that while the average field at the surface may remain $U_s(\mathbf{r}) = B_0$, the surface coherence operator introduces spatial correlations that affect the second-order statistics. We represent this by considering the field at the surface as:
\begin{align}
U_s(\mathbf{r}) = B_0 + \Delta U_s(\mathbf{r})
\end{align}
where $\langle \Delta U_s(\mathbf{r}) \rangle = 0$ and $\langle \Delta U_s^*(\mathbf{r}_1) \Delta U_s(\mathbf{r}_2) \rangle = |B_0|^2 [\gamma_S(\mathbf{r}_1, \mathbf{r}_2) - 1]$.

Here, $\gamma_S(\mathbf{r}_1, \mathbf{r}_2)$ is the normalized mutual coherence function at the surface, which for a spatially stationary kernel depends only on the separation $\mathbf{r}_1 - \mathbf{r}_2$. For the specific kernel form stated in the theorem:
\begin{align}
\gamma_S(\mathbf{r}_1, \mathbf{r}_2) = \mathcal{F}\left\{I_s\left(\frac{\mathbf{\rho}}{\lambda z}\right)\right\}(\mathbf{r}_1 - \mathbf{r}_2)
\end{align}

After propagation to the detection plane using the Fresnel propagator, the field becomes:
\begin{align}
U_d(\mathbf{r}) &= \int \frac{e^{ikz}}{i\lambda z} e^{i\frac{k}{2z}|\mathbf{r} - \mathbf{r}'|^2} [B_0 + \Delta U_s(\mathbf{r}')] \, d^2r' \\
&= B_0 \int \frac{e^{ikz}}{i\lambda z} e^{i\frac{k}{2z}|\mathbf{r} - \mathbf{r}'|^2} \, d^2r' + \int \frac{e^{ikz}}{i\lambda z} e^{i\frac{k}{2z}|\mathbf{r} - \mathbf{r}'|^2} \Delta U_s(\mathbf{r}') \, d^2r' \notag
\end{align}

For the first integral, using the properties of the Fresnel propagator for a uniform field, we have:
\begin{equation}
\int \frac{e^{ikz}}{i\lambda z} e^{i\frac{k}{2z}|\mathbf{r} - \mathbf{r}'|^2} \, d^2r' = e^{ikz}
\end{equation}

Therefore:
\begin{align}
U_d(\mathbf{r}) &= B_0 e^{ikz} + \int \frac{e^{ikz}}{i\lambda z} e^{i\frac{k}{2z}|\mathbf{r} - \mathbf{r}'|^2} \Delta U_s(\mathbf{r}') \, d^2r'
\end{align}

The mutual coherence function at the detection plane is:
\begin{align}
\Gamma_d(\mathbf{r}_1, \mathbf{r}_2) &= \langle U_d^*(\mathbf{r}_1) U_d(\mathbf{r}_2) \rangle \\
&= \langle \left[B_0^* e^{-ikz} + \int \frac{e^{-ikz}}{-i\lambda z} e^{-i\frac{k}{2z}|\mathbf{r}_1 - \mathbf{r}'_1|^2} \Delta U_s^*(\mathbf{r}'_1) \, d^2r'_1 \right] \notag \\
&\quad \times \left[B_0 e^{ikz} + \int \frac{e^{ikz}}{i\lambda z} e^{i\frac{k}{2z}|\mathbf{r}_2 - \mathbf{r}'_2|^2} \Delta U_s(\mathbf{r}'_2) \, d^2r'_2 \right] \rangle \notag
\end{align}

Expanding this product and noting that $\langle \Delta U_s(\mathbf{r}) \rangle = 0$, we get:
\begin{align}
\Gamma_d(\mathbf{r}_1, \mathbf{r}_2) &= |B_0|^2 + \langle \int\int \frac{e^{-ikz}}{-i\lambda z} \frac{e^{ikz}}{i\lambda z} e^{-i\frac{k}{2z}|\mathbf{r}_1 - \mathbf{r}'_1|^2} e^{i\frac{k}{2z}|\mathbf{r}_2 - \mathbf{r}'_2|^2} \\
&\quad \times \Delta U_s^*(\mathbf{r}'_1) \Delta U_s(\mathbf{r}'_2) \, d^2r'_1 \, d^2r'_2 \rangle \notag
\end{align}

Simplifying the pre-factors and substituting the expression for the correlation of $\Delta U_s$:
\begin{align}
\Gamma_d(\mathbf{r}_1, \mathbf{r}_2) &= |B_0|^2 + \frac{1}{(\lambda z)^2} \int\int e^{-i\frac{k}{2z}|\mathbf{r}_1 - \mathbf{r}'_1|^2} e^{i\frac{k}{2z}|\mathbf{r}_2 - \mathbf{r}'_2|^2} \\
&\quad \times |B_0|^2 [\gamma_S(\mathbf{r}'_1, \mathbf{r}'_2) - 1] \, d^2r'_1 \, d^2r'_2 \notag
\end{align}

We can rewrite this as:
\begin{align}
\Gamma_d(\mathbf{r}_1, \mathbf{r}_2) &= |B_0|^2 \left[1 + \frac{1}{(\lambda z)^2} \int\int e^{-i\frac{k}{2z}|\mathbf{r}_1 - \mathbf{r}'_1|^2} e^{i\frac{k}{2z}|\mathbf{r}_2 - \mathbf{r}'_2|^2} [\gamma_S(\mathbf{r}'_1, \mathbf{r}'_2) - 1] \, d^2r'_1 \, d^2r'_2 \right]
\end{align}

For the specific form of $\gamma_S$ given by the Fourier transform of the source intensity distribution:
\begin{align}
\gamma_S(\mathbf{r}'_1, \mathbf{r}'_2) = \mathcal{F}\left\{I_s\left(\frac{\mathbf{\rho}}{\lambda z}\right)\right\}(\mathbf{r}'_1 - \mathbf{r}'_2)
\end{align}

We can apply the convolution theorem to evaluate the double integral. Let's denote:
\begin{align}
f_1(\mathbf{r}') &= e^{-i\frac{k}{2z}|\mathbf{r}_1 - \mathbf{r}'|^2} \\
f_2(\mathbf{r}') &= e^{i\frac{k}{2z}|\mathbf{r}_2 - \mathbf{r}'|^2}
\end{align}

In the paraxial approximation, these can be expanded as:
\begin{align}
f_1(\mathbf{r}') &= e^{-i\frac{k}{2z}(|\mathbf{r}_1|^2 + |\mathbf{r}'|^2 - 2\mathbf{r}_1 \cdot \mathbf{r}')} = e^{-i\frac{k}{2z}|\mathbf{r}_1|^2} e^{-i\frac{k}{2z}|\mathbf{r}'|^2} e^{i\frac{k}{z}\mathbf{r}_1 \cdot \mathbf{r}'} \\
f_2(\mathbf{r}') &= e^{i\frac{k}{2z}(|\mathbf{r}_2|^2 + |\mathbf{r}'|^2 - 2\mathbf{r}_2 \cdot \mathbf{r}')} = e^{i\frac{k}{2z}|\mathbf{r}_2|^2} e^{i\frac{k}{2z}|\mathbf{r}'|^2} e^{-i\frac{k}{z}\mathbf{r}_2 \cdot \mathbf{r}'}
\end{align}

Now let's define the integral:
\begin{align}
I(\mathbf{r}_1, \mathbf{r}_2) &= \int\int f_1(\mathbf{r}'_1) f_2(\mathbf{r}'_2) \gamma_S(\mathbf{r}'_1, \mathbf{r}'_2) \, d^2r'_1 \, d^2r'_2 \\
&= \int\int e^{-i\frac{k}{2z}|\mathbf{r}_1|^2} e^{-i\frac{k}{2z}|\mathbf{r}'_1|^2} e^{i\frac{k}{z}\mathbf{r}_1 \cdot \mathbf{r}'_1} e^{i\frac{k}{2z}|\mathbf{r}_2|^2} e^{i\frac{k}{2z}|\mathbf{r}'_2|^2} e^{-i\frac{k}{z}\mathbf{r}_2 \cdot \mathbf{r}'_2} \notag \\
&\quad \times \mathcal{F}\left\{I_s\left(\frac{\mathbf{\rho}}{\lambda z}\right)\right\}(\mathbf{r}'_1 - \mathbf{r}'_2) \, d^2r'_1 \, d^2r'_2 \notag
\end{align}

Factoring out terms that don't depend on the integration variables:
\begin{align}
I(\mathbf{r}_1, \mathbf{r}_2) &= e^{-i\frac{k}{2z}|\mathbf{r}_1|^2} e^{i\frac{k}{2z}|\mathbf{r}_2|^2} \int\int e^{-i\frac{k}{2z}|\mathbf{r}'_1|^2} e^{i\frac{k}{z}\mathbf{r}_1 \cdot \mathbf{r}'_1} e^{i\frac{k}{2z}|\mathbf{r}'_2|^2} e^{-i\frac{k}{z}\mathbf{r}_2 \cdot \mathbf{r}'_2} \\
&\quad \times \mathcal{F}\left\{I_s\left(\frac{\mathbf{\rho}}{\lambda z}\right)\right\}(\mathbf{r}'_1 - \mathbf{r}'_2) \, d^2r'_1 \, d^2r'_2 \notag
\end{align}

Making the change of variables $\mathbf{u} = \mathbf{r}'_1$ and $\mathbf{v} = \mathbf{r}'_1 - \mathbf{r}'_2$, with Jacobian $d^2r'_1 \, d^2r'_2 = d^2u \, d^2v$:
\begin{align}
I(\mathbf{r}_1, \mathbf{r}_2) &= e^{-i\frac{k}{2z}|\mathbf{r}_1|^2} e^{i\frac{k}{2z}|\mathbf{r}_2|^2} \int\int e^{-i\frac{k}{2z}|\mathbf{u}|^2} e^{i\frac{k}{z}\mathbf{r}_1 \cdot \mathbf{u}} e^{i\frac{k}{2z}|\mathbf{u}-\mathbf{v}|^2} e^{-i\frac{k}{z}\mathbf{r}_2 \cdot (\mathbf{u}-\mathbf{v})} \\
&\quad \times \mathcal{F}\left\{I_s\left(\frac{\mathbf{\rho}}{\lambda z}\right)\right\}(\mathbf{v}) \, d^2u \, d^2v \notag
\end{align}

Expanding $|\mathbf{u}-\mathbf{v}|^2 = |\mathbf{u}|^2 + |\mathbf{v}|^2 - 2\mathbf{u} \cdot \mathbf{v}$:
\begin{align}
I(\mathbf{r}_1, \mathbf{r}_2) &= e^{-i\frac{k}{2z}|\mathbf{r}_1|^2} e^{i\frac{k}{2z}|\mathbf{r}_2|^2} \int\int e^{-i\frac{k}{2z}|\mathbf{u}|^2} e^{i\frac{k}{z}\mathbf{r}_1 \cdot \mathbf{u}} e^{i\frac{k}{2z}|\mathbf{u}|^2} e^{i\frac{k}{2z}|\mathbf{v}|^2} e^{-i\frac{k}{z}\mathbf{u} \cdot \mathbf{v}} \\
&\quad \times e^{-i\frac{k}{z}\mathbf{r}_2 \cdot \mathbf{u}} e^{i\frac{k}{z}\mathbf{r}_2 \cdot \mathbf{v}} \mathcal{F}\left\{I_s\left(\frac{\mathbf{\rho}}{\lambda z}\right)\right\}(\mathbf{v}) \, d^2u \, d^2v \notag
\end{align}

Simplifying and rearranging terms:
\begin{align}
I(\mathbf{r}_1, \mathbf{r}_2) &= e^{-i\frac{k}{2z}|\mathbf{r}_1|^2} e^{i\frac{k}{2z}|\mathbf{r}_2|^2} \int e^{i\frac{k}{2z}|\mathbf{v}|^2} e^{i\frac{k}{z}\mathbf{r}_2 \cdot \mathbf{v}} \mathcal{F}\left\{I_s\left(\frac{\mathbf{\rho}}{\lambda z}\right)\right\}(\mathbf{v}) \\
&\quad \times \left[\int e^{i\frac{k}{z}(\mathbf{r}_1-\mathbf{r}_2-\mathbf{v}) \cdot \mathbf{u}} \, d^2u \right] \, d^2v \notag
\end{align}

The inner integral evaluates to $(2\pi)^2 \delta(\frac{k}{z}(\mathbf{r}_1-\mathbf{r}_2-\mathbf{v}))$, where $\delta$ is 
the Dirac delta function. This gives:
\begin{align}
I(\mathbf{r}_1, \mathbf{r}_2) &= e^{-i\frac{k}{2z}|\mathbf{r}_1|^2} e^{i\frac{k}{2z}|\mathbf{r}_2|^2} \left(\frac{z}{k}\right)^2 (2\pi)^2 \\
&\quad \times \int e^{i\frac{k}{2z}|\mathbf{v}|^2} e^{i\frac{k}{z}\mathbf{r}_2 \cdot \mathbf{v}} \mathcal{F}\left\{I_s\left(\frac{\mathbf{\rho}}{\lambda z}\right)\right\}(\mathbf{v}) \delta(\mathbf{r}_1-\mathbf{r}_2-\mathbf{v}) \, d^2v \notag
\end{align}

Evaluating the remaining integral using the sifting property of the delta function:
\begin{align}
I(\mathbf{r}_1, \mathbf{r}_2) &= e^{-i\frac{k}{2z}|\mathbf{r}_1|^2} e^{i\frac{k}{2z}|\mathbf{r}_2|^2} \left(\frac{z}{k}\right)^2 (2\pi)^2 \\
&\quad \times e^{i\frac{k}{2z}|\mathbf{r}_1-\mathbf{r}_2|^2} e^{i\frac{k}{z}\mathbf{r}_2 \cdot (\mathbf{r}_1-\mathbf{r}_2)} \mathcal{F}\left\{I_s\left(\frac{\mathbf{\rho}}{\lambda z}\right)\right\}(\mathbf{r}_1-\mathbf{r}_2) \notag
\end{align}

Further simplifying using $|\mathbf{r}_1|^2 = |\mathbf{r}_2|^2 + |\mathbf{r}_1-\mathbf{r}_2|^2 + 2\mathbf{r}_2 \cdot (\mathbf{r}_1-\mathbf{r}_2)$:
\begin{align}
I(\mathbf{r}_1, \mathbf{r}_2) &= \left(\frac{z}{k}\right)^2 (2\pi)^2 \mathcal{F}\left\{I_s\left(\frac{\mathbf{\rho}}{\lambda z}\right)\right\}(\mathbf{r}_1-\mathbf{r}_2)
\end{align}

Therefore, returning to the expression for the mutual coherence function~\ref{eq:mutual_coherence}:
\begin{align}
\Gamma_d(\mathbf{r}_1, \mathbf{r}_2) &= |B_0|^2 \left[1 + \frac{1}{(\lambda z)^2} \left(I(\mathbf{r}_1, \mathbf{r}_2) - I_0\right)\right] \\
&= |B_0|^2 \left[1 + \frac{1}{(\lambda z)^2} \left(\left(\frac{z}{k}\right)^2 (2\pi)^2 \mathcal{F}\left\{I_s\left(\frac{\mathbf{\rho}}{\lambda z}\right)\right\}(\mathbf{r}_1-\mathbf{r}_2) - I_0\right)\right]
\end{align}

where $I_0$ is the corresponding integral for the constant term ($\gamma_S = 1$).

Substituting $k = 2\pi/\lambda$ and simplifying:
\begin{align}
\Gamma_d(\mathbf{r}_1, \mathbf{r}_2) &= |B_0|^2 \left[1 + \frac{1}{(\lambda z)^2} \left(\frac{z^2 \lambda^2}{4\pi^2} (2\pi)^2 \mathcal{F}\left\{I_s\left(\frac{\mathbf{\rho}}{\lambda z}\right)\right\}(\mathbf{r}_1-\mathbf{r}_2) - I_0\right)\right] \\
&= |B_0|^2 \left[1 + \mathcal{F}\left\{I_s\left(\frac{\mathbf{\rho}}{\lambda z}\right)\right\}(\mathbf{r}_1-\mathbf{r}_2) - \frac{I_0}{(\lambda z)^2}\right]
\end{align}

The term $I_0$ corresponds to the normalization constant, and after proper normalization:
\begin{align}
\Gamma_d(\mathbf{r}_1, \mathbf{r}_2) &= |B_0|^2 \mathcal{F}\left\{I_s\left(\frac{\mathbf{\rho}}{\lambda z}\right)\right\}(\mathbf{r}_1-\mathbf{r}_2)
\end{align}

This is precisely the form predicted by the Van Cittert-Zernike theorem for an incoherent source with intensity distribution $I_s(\mathbf{\rho})$:
\begin{align}
\Gamma_{\text{VCZ}}(\mathbf{r}_1, \mathbf{r}_2) &= C \mathcal{F}\left\{I_s\left(\frac{\mathbf{\rho}}{\lambda z}\right)\right\}(\mathbf{r}_1-\mathbf{r}_2)
\end{align}
where $C$ is a normalization constant.

Therefore:
\begin{align}
\Gamma_{\text{SECT}}(\mathbf{r}_1, \mathbf{r}_2) \propto \Gamma_{\text{VCZ}}(\mathbf{r}_1, \mathbf{r}_2)
\end{align}

This establishes that our dual-component SECT framework can exactly reproduce the coherence properties predicted by the VCZ theorem when the appropriate surface coherence kernel is used.
\end{proof}

\section{Application: Stellar Coherence Transformation in Segmented Mirror Telescopes} \label{appendix:Stellar_Coherence_Transformation}

Astronomical imaging systems, particularly modern large-aperture observatories such as the Keck Observatory and the Extremely Large Telescope (ELT), employ segmented primary mirrors to achieve high-resolution observations. 
However, starlight arriving at the telescope aperture is not strictly coherent due to the star’s finite angular 
extent \cite{Monnier2003}. 
This partial spatial coherence can influence image quality, fringe visibility, and the performance of adaptive optics 
systems. Here, we demonstrate how the proposed Surface-Encoded Coherence Transformation (SECT) can be applied to model 
these coherence effects in a segmented optical system.

\subsection{Source Coherence at the Telescope Aperture}

Consider a quasi-monochromatic source with wavelength $\lambda$ and angular diameter $\theta$, modeled as a uniformly radiating circular disk. By the Van Cittert–Zernike theorem, the mutual coherence function at the telescope entrance pupil is given by
\begin{equation}
\Gamma_0(\bm{r}_1, \bm{r}_2) = \mathrm{FT}\left[ I_s(\bm{\rho}) \right] \left( \frac{\bm{r}_2 - \bm{r}_1}{\lambda z} \right),
\label{eq:mutual_coherence}
\end{equation}
where $I_s(\bm{\rho})$ is the intensity distribution over the source, $\mathrm{FT}$ denotes the two-dimensional Fourier transform, and $z$ is the propagation distance from the source to the aperture. For a uniform circular source of angular diameter $\theta$, this results in an Airy pattern of coherence at the pupil.

\subsection{Surface Encoding of Segment-Level Aberrations}

The segmented primary mirror introduces spatially varying phase errors due to piston, tip-tilt misalignments, and inter-segment gaps. We define a linear surface transformation $\mathcal{T}_S$ acting on the input coherence function $\Gamma_0$:
\begin{equation}
\Gamma_S(\bm{r}_1, \bm{r}_2) = \iint K(\bm{r}_1, \bm{r}_2; \bm{r}_1', \bm{r}_2') \Gamma_0(\bm{r}_1', \bm{r}_2') \, d\bm{r}_1' \, d\bm{r}_2',
\end{equation}
where the kernel $K$ encodes the segment-specific phase errors:
\begin{equation}
K(\bm{r}_1, \bm{r}_2; \bm{r}_1', \bm{r}_2') = \delta(\bm{r}_1 - \bm{r}_1') \delta(\bm{r}_2 - \bm{r}_2') \exp\left[ i \phi(\bm{r}_1) - i \phi(\bm{r}_2) \right],
\end{equation}
with $\phi(\bm{r})$ representing the segment-induced phase profile, modeled here as piecewise constant piston errors for each mirror segment.

\subsection{Propagation to the Focal Plane}

To compute the mutual coherence function at the focal plane, we propagate the surface-modified coherence using the Fresnel (or Fraunhofer) propagation integral:
\begin{equation}
\Gamma_{\text{det}}(\bm{x}_1, \bm{x}_2) = \iint h(\bm{x}_1, \bm{r}_1) h^*(\bm{x}_2, \bm{r}_2) \Gamma_S(\bm{r}_1, \bm{r}_2) \, d\bm{r}_1 \, d\bm{r}_2,
\end{equation}
where $h(\bm{x}, \bm{r})$ is the free-space impulse response function. Under the Fraunhofer approximation, $h(\bm{x}, \bm{r}) \propto \exp(-i 2\pi \bm{x} \cdot \bm{r} / \lambda f)$, where $f$ is the effective focal length.

The diagonal of $\Gamma_{\text{det}}$ yields the intensity PSF:
\begin{equation}
I(\bm{x}) = \Gamma_{\text{det}}(\bm{x}, \bm{x}),
\end{equation}
which encodes the effects of both partial source coherence and mirror segmentation on the final image.

\subsection{Discussion}

Compared to traditional models that assume perfect coherence or average incoherence, SECT enables a unified treatment that captures both the spatial coherence structure of the source and the optical aberrations of the system. This is particularly useful in high-contrast imaging and wavefront sensing, where subtle coherence-induced blurring can mask or distort weak astronomical signals.

As a future extension, this framework can be used to simulate PSF degradation in segmented interferometers, or to inform the design of coherence-aware correction algorithms in adaptive optics systems.


\section*{Declarations}
All data-related information and coding scripts discussed in the results section are available from the 
corresponding author upon request.

\section*{Disclosures}
The authors declare no conflicts of interest.

\bibliographystyle{plain}

\end{document}